\documentclass[11pt]{article}

\usepackage[utf8]{inputenc}
\usepackage[T1]{fontenc}
\usepackage{lmodern}
\usepackage[letterpaper,margin=1in]{geometry}
\usepackage{amsmath,amssymb,amsthm,amsfonts,mathtools}
\usepackage{bm}
\usepackage[table]{xcolor}
\usepackage{array,longtable}
\usepackage{hyperref}
\usepackage{aliascnt}
\usepackage[capitalise]{cleveref}

\newtheorem{theorem}{Theorem}[section]
\newtheorem{theoremA}{Theorem}

\newaliascnt{lemma}{theorem}
\newtheorem{lemma}[lemma]{Lemma}
\aliascntresetthe{lemma}

\newaliascnt{corollary}{theorem}
\newtheorem{corollary}[corollary]{Corollary}
\aliascntresetthe{corollary}

\newaliascnt{proposition}{theorem}
\newtheorem{proposition}[proposition]{Proposition}
\aliascntresetthe{proposition}

\newtheorem{remark}[theorem]{Remark}

\crefname{theorem}{theorem}{theorems}
\Crefname{theorem}{Theorem}{Theorems}
\crefname{lemma}{lemma}{lemmas}
\Crefname{lemma}{Lemma}{Lemmas}
\crefname{proposition}{proposition}{propositions}
\Crefname{proposition}{Proposition}{Propositions}
\crefname{corollary}{corollary}{corollaries}
\Crefname{corollary}{Corollary}{Corollaries}
\crefname{definition}{definition}{definitions}
\Crefname{definition}{Definition}{Definitions}
\crefname{remark}{remark}{remarks}
\Crefname{remark}{Remark}{Remarks}

\DeclareMathOperator{\wt}{wt}
\DeclareMathOperator{\vars}{vars}

\newcommand{\abs}[1]{\lvert #1 \rvert}
\newcommand{\Ball}[2]{B(#1,#2)}

\newcommand{\indA}{{\!\scriptscriptstyle A}}
\newcommand{\mDelta}{\Gamma}
\newcommand{\indB}{{\!\scriptscriptstyle B}}
\newcommand{\av}{\hspace{0.05em}\mathrm{light}}
\newcommand{\LightBall}[3]{\mathcal{B}_{#1}(#2,#3)}

\newcolumntype{L}[1]{>{\raggedright\arraybackslash}p{#1}}

\begin{document}

\title{Dequantizing Short-Path Quantum Algorithms}
\date{}
    \author{
    Fran{\c c}ois Le Gall\thanks{Graduate School of Mathematics, Nagoya University.}
    \and
    Suguru Tamaki\thanks{Graduate School of Information Science, University of Hyogo.}
 }     

\pagenumbering{roman}

\maketitle

% ==================================================
\begin{abstract}
The short-path quantum algorithm introduced by Hastings (Quantum 2018, 2019) is a variant of adiabatic quantum algorithms that enables an easier worst-case analysis by avoiding the need to control the spectral gap along a long adiabatic path. Dalzell, Pancotti, Campbell, and Brand\~{a}o (STOC 2023) recently revisited this framework and obtained a clear analysis of the complexity of the short-path algorithm for several classes of constraint satisfaction problems (MAX-$k$-CSPs), leading to quantum algorithms with complexity $2^{(1-c)n/2}$ for some constant $c>0$. This suggested a super-quadratic quantum advantage over classical algorithms.

In this work, we identify an explicit classical mechanism underlying a substantial part of this line of work, and show that it leads to clean dequantizations. As a consequence, we obtain classical algorithms that run in time $2^{(1-c')n}$, for some constant $c'>c$, for the same classes of constraint satisfaction problems. This shows that current short-path quantum algorithms for these problems do not achieve a super-quadratic advantage. On the positive side, our results provide a new ``quantum-inspired'' approach to designing classical algorithms for important classes of constraint satisfaction problems.
\end{abstract}

\newpage
\tableofcontents
\newpage

\pagenumbering{arabic}

% ==================================================
% Section 1: Introduction
% ==================================================
\section{Introduction}\label{sec:introduction}
\subsection{Background}
Grover search \cite{Grover96} and amplitude amplification \cite{BrassardHMT02} provide the standard quantum baseline for unstructured search and exact optimization. Relative to exhaustive search over $2^n$ candidates, this baseline yields a quantum running time of $2^{n/2}$, up to polynomial factors. A central question in quantum optimization is therefore to identify structural mechanisms that achieve a running time of the form
\[ 
2^{(1-c)n/2}
\]
for some explicit constant $c>0$.

\paragraph{Short-path quantum algorithms.}
One influential line of work in this direction is the short-path quantum algorithm introduced by Hastings~\cite{Hastings18video,Hastings18short,Hastings18weaker,Hastings19duality,Hastings2019short}. 

Hastings’ work is inspired by adiabatic quantum algorithms (see, e.g., \cite{AlbashLidar2018,Altshuler2010,Farhi2000,HenYoung2011,Wecker2016,YoungKnyshSmelyanskiy2010}).
Adiabatic quantum algorithms prepare a quantum system in the ground state of a simple initial Hamiltonian $H_0$, and then slowly evolve it into a problem Hamiltonian $H_1$ whose ground state encodes the solution to the optimization problem. By the adiabatic theorem \cite{AlbashLidar2018,Farhi2000}, if the evolution is sufficiently slow (relative to the inverse square of the minimum spectral gap), the system remains in the ground state throughout the evolution, so that measuring at the end yields the desired solution. When analyzing the performance of an adiabatic quantum algorithm, the main challenge is to establish a lower bound on the spectral gap that holds throughout the entire evolution. 

Hastings’ short-path algorithm avoids the need to control the spectral gap along a long adiabatic path by instead evolving only for a short time and then ``jumping to the end’’ via measurement. The evolution is designed so as to amplify the amplitude on low-energy states of the problem Hamiltonian, so that a measurement (possibly combined with repetition or amplitude amplification) yields a good solution with non-negligible probability.

\paragraph{Short-path quantum algorithms for MAX-E$\boldsymbol{k}$-LIN2.}
Hastings’ main work \cite{Hastings18short,Hastings18weaker,Hastings19duality} focuses on the optimization problem MAX-E\(k\)-LIN2 (Maximum Exact-$k$ Linear Equations mod~2) for $k\ge 2$. This problem asks to minimize the value of a degree-$k$ polynomial 
\[ 
    H(x) = p(x_1,\ldots,x_n) 
\]
over $x\in\{-1,1\}^n$, where $p$ consists only of monomials of degree $k$ involving exactly~$k$ variables (the coefficients can be arbitrary real numbers). Ref.~\cite{Hastings18short} showed that, under certain conditions, the short-path quantum algorithm achieves a running time of the form $2^{(1-c)n/2}$ for some constant $c>0$. The main conditions are (i) a condition on the density of low-energy states and (ii) a condition on the degeneracy of the ground state. Hastings later showed that the latter condition can be removed \cite{Hastings18weaker}, leaving only the low-energy-density condition.

Dalzell, Pancotti, Campbell, and Brandão~\cite{DalzellPCB23} recently revisited this framework and obtained a clean characterization of the complexity of the short-path algorithm for MAX-E$k$-LIN2. In particular, they reformulated the low-energy-density condition as follows. For some parameter $\gamma \in (0,1)$, consider the threshold set
\[ 
    T_\eta = \{x : H(x) \le (1-\eta) H_{\mathrm{min}}\},
\]
where $H_{\mathrm{min}} \le 0$ denotes the optimum value. The low-energy-density condition posits that there exist constants $\gamma,\eta \in (0,1)$ such that
\begin{equation}\label{eq:cond}
|T_\eta| \le 2^{(1-\gamma)n}\tag{$\star$}\,.
\end{equation}
Ref.~\cite{DalzellPCB23} shows that, under this condition, the short-path framework\footnote{While the algorithm from \cite{DalzellPCB23} and its analysis are inspired by \cite{Hastings18short}, the algorithms are not identical. Informally, the algorithm from \cite{DalzellPCB23} makes a small jump and then a large jump, while the algorithms from \cite{Hastings18short,Hastings18weaker} make a large jump and then a small jump.} yields a quantum algorithm with running time of the form $2^{(1-c)n/2}$ for some constant $c>0$ (the formal statement appears as \Cref{th6} below).

\paragraph{Short-path quantum algorithms for MAX-\(\boldsymbol{k}\)-CSP.}
The second main result in \cite{DalzellPCB23} shows how to remove Condition~(\ref{eq:cond}). This comes at the cost of introducing new quantities in the complexity (details are given later). This result applies not only to MAX-E$k$-LIN2, but also to the more general class of problems called MAX-$k$-CSP, which we now define in the framework of Hamiltonian minimization.\footnote{
Although all the results in the present paper are established for the more general weighted MAX-$k$-CSP, in this introduction we consider, for simplicity, only the unweighted version, as in \cite{DalzellPCB23}.}

An instance of (unweighted) MAX-$k$-CSP is defined by the Hamiltonian
\[ 
H(x) = \sum_{j=1}^m \mathcal{C}_j(x),
\]
where each $\mathcal{C}_j$ is a constraint satisfying the following properties: 
\begin{itemize}
\item[(i)]
it depends on $k_j$ of the $n$ variables, for some $k_j\in\{1,\ldots,k\}$;
%(we denote by $k_j \le k$ the number of bits on which it depends);
\item[(ii)]
it takes the value $-1$ on $s_j$ assignments (the ``satisfying assignments'') of these $k_j$ variables, for some $s_j\in \{1,\ldots,2^{k_j} - 1\}$;
\item[(iii)] it takes the value $s_j/(2^{k_j} - s_j)$ on the remaining $2^{k_j} - s_j$ assignments (the ``unsatisfying assignments''), so that the average value of the constraint over all $2^{k_j}$ assignments is zero.
\end{itemize}
% (i) it depends on at most $k$ of the $n$ bits of the input $x$ (we denote by $k_j \le k$ the number of bits on which it depends), and (ii) it takes the value $-1$ on $s_j \in [1,2^{k_j} - 1]$ (“satisfying”) assignments to those bits and the value $s_j/(2^{k_j} - s_j)$ on the remaining $2^{k_j} - s_j$ (“unsatisfying”) assignments, so that the average of each term over all $2^{k_j}$ assignments of these bits is zero. 
We again denote by $H_{\mathrm{min}}$ the minimum value of $H(x)$ over all assignments $x$. 

For example, for MAX-E3-SAT (the version of MAX-3-SAT in which each clause contains exactly 3 variables), each constraint takes value $-1$ on the 7 satisfying assignments of the three variables, and value $7$ on the unsatisfying assignment of the three variables. Thus, we have $H_{\mathrm{min}}=7m-8m^\ast$, where $m^\ast$ denotes the maximum number of satisfiable clauses in the formula. Since $m^\ast\ge 7m/8$ (as can be seen by considering a random assignment), we have $H_{\mathrm{min}}\in[-m,0]$.

Ref.~\cite{DalzellPCB23} shows that the complexity of the short-path quantum algorithm depends on two quantities: one representing the regularity of the instance, and another representing how small $H_{\mathrm{min}}$ is (i.e., in the case of MAX-E3-SAT, how large $m^\ast$ is). For sufficiently regular instances with sufficiently small 
$H_{\mathrm{min}}$—which is arguably the most interesting case—MAX-$k$-CSP can be solved by a quantum algorithm in time $2^{(1-c)n/2}$ for some constant $c>0$ (the formal statement appears as \Cref{th7} below).

\paragraph{Main question.}
These recent results demonstrate the potential of ``super-quadratic’’ speedups by the short-path quantum algorithm for fundamental classes of constraint satisfaction problems. To substantiate this claim, however, the following question needs to be answered:

\begin{center}
\emph{What is the performance of classical algorithms under the same conditions?}
\end{center}

\subsection{Statement of our results}\label{sub:results}
In this work, we identify a classical mechanism underlying the analysis of short-path quantum algorithms in \cite{DalzellPCB23,Hastings18short,Hastings18weaker,Hastings19duality}, and show that these quantum algorithms admit clean dequantizations.

As a consequence, we obtain classical algorithms that run in time $2^{(1-c)n}$ for some constant $c > 0$ under the same conditions as the quantum algorithms in \cite{DalzellPCB23}.\footnote{The same conclusion also applies under the conditions used in the analysis of the short-path quantum algorithm for MAX-E\(k\)-LIN2 carried out in Hastings’ work~\cite{Hastings18short,Hastings18weaker}.
%, which preceded the algorithm of~\cite{DalzellPCB23}. 
It is nevertheless difficult to compare our constant~\(c\) with the corresponding constants in~\cite{Hastings18short,Hastings18weaker}, since these are not given in fully explicit form.} Furthermore, our constant $c$ improves upon the corresponding constant in \cite{DalzellPCB23}. This shows that current short-path quantum algorithms for constraint satisfaction problems do not achieve a super-quadratic advantage. On the positive side, our results provide a new ``quantum-inspired’’ approach to designing classical algorithms for important classes of constraint satisfaction problems.

%======
\paragraph{Our first main result: classical algorithm for MAX-E$\boldsymbol{k}$-LIN2.}
%======
The formal statement of the first main result from \cite{DalzellPCB23} is as follows.\footnote{In this paper the notation $O^\ast(\cdot)$ removes $2^{o(n)}$ factors.} 
\begin{theoremA}[Theorem 6 in \cite{DalzellPCB23}]\label{th6}
Assume that $k\le n/6$.
For any parameters $\gamma\in[(1+4\log_2 n)/n,1]$ and $\eta\in[0,1]$,
there is a quantum algorithm which, for any instance of MAX-E\(k\)-LIN2 such that Condition \eqref{eq:cond} holds, produces an optimal solution with high probability in time
\[
O^\ast\left(2^{(1-c^{\mathrm{q}}_\indA)n/2}\right)\,,\,\,\textrm{where }\,\,
c^{\mathrm{q}}_\indA=\frac{1}{2(2+\ln 2)}\frac{\gamma\eta}{k}\approx 0.1856\frac{\gamma\eta}{k}.
\]
\end{theoremA}

We obtain the following result.

\begin{theorem}[Simplified version of \Cref{th:concrete-eklin2}]\label{th6c}
%For any parameters $\gamma\in[(1+4\log_2 n)/n,1]$ and $\eta\in[0,1]$,
%there is a classical algorithm which, for any instance of MAX-E\(k\)-LIN2 with $n\ge 6k$ such that \cref{eq:cond} holds, 
Under the same conditions as in \Cref{th6}, there is a classical algorithm which, for any instance of MAX-E\(k\)-LIN2 such that Condition \eqref{eq:cond} holds, 
produces an optimal solution with high probability in time
\[
O^\ast\left(2^{(1-c^{\mathrm{cl}}_\indA)n}\right)\,,\,\,\textrm{where }\,\,
c^{\mathrm{cl}}_\indA\ge \min\!\left\{\gamma,h\!\left(\frac{\eta}{2k}\right)\right\}\,,
\]
where $h(\cdot)$ denotes the binary entropy.
\end{theorem}
Note that since $\min\{a,b\}\ge ab$ for any $a,b\in[0,1]$ and 
since $h(x)\ge 2x$ for \(x\in[0,1/2]\), we have
\[
\min\!\left\{\gamma,h\!\left(\frac{\eta}{2k}\right)\right\}
\ge
\gamma\,h\!\left(\frac{\eta}{2k}\right)
\ge
\frac{\gamma\eta}{k}.
\]
In addition to having a larger leading constant, $c^{\mathrm{cl}}_\indA$ can be asymptotically larger than $c^{\mathrm{q}}_\indA$
since its dependence on $\gamma\eta/k$
is entropic instead of linear. Indeed, 
\[
h(x)\sim x\log_2\!\frac1x
\qquad (x\to 0).
\]

\paragraph{Our second main result: classical algorithm for  MAX-$\boldsymbol{k}$-CSP.}
We first give the formal statement obtained for MAX-$k$-CSP in \cite{DalzellPCB23}. 
%We first give the formal statement obtained for MAX-$k$-CSPs in \cite{DalzellPCB23}.
The complexity is expressed using two quantities that represent properties of the input instance. The first parameter, denoted $D$, represents the \emph{irregularity} of the instance:
\[
    D = \frac{n \sum_{j=1}^n d_j^2}{k^2 m^2}\,, 
\]
where $d_j$ is the number of constraints that involve variable $j$. Note that $D=1$ for a perfectly balanced instance (since in this case $d_j=km/n$ for each $j$). The second parameter, which we call in this paper the \emph{normalized optimum scale} of the instance, is
\[
\Delta = \frac{|H_{\mathrm{min}}|}{m}\,.
\]
For instance, for MAX-E3-SAT we have $\Delta=(8m^\ast-7m)/m$.
Note that we have $\Delta=1$ if all constraints can be satisfied, and $\Delta=0$ if the optimal value is exactly the same as the expected value of a random assignment. 
The parameter $\Delta$ thus represents the (normalized) gap between the minimum value and the expected value of a random assignment.
%For instance, for MAX-E3-SAT we have $\Delta=(7m-8m^\ast)/m$, which represents the (normalized) gap between the maximal number of satisfiable clauses and the value $7m/8$ corresponding to a random assignment.

\begin{theoremA}[Theorem 7 in \cite{DalzellPCB23}]\label{th7}
There is a quantum algorithm which, for any instance of unweighted MAX-\(k\)-CSP, produces an optimal solution with high probability in time
\[
O^\ast\left(2^{(1-c^{\mathrm{q}}_\indB)n/2}\right)\,,\,\,\textrm{where }\,\,
c^{\mathrm{q}}_\indB
=
% 0.0578\,
% \frac{1}{2^{3k}k^3D}
% \left(\frac{|H_{\mathrm{min}}|}{m}\right)^3\,.
0.0578\cdot\frac{\Delta^3}{2^{3k}k^3D}\,.
\]
\end{theoremA}
%\Cref{th6} and \Cref{th7} showed the potential of super-Grover quantum speedup.

One of the main results of our paper implies the following result:
\begin{theorem}[Simplified version of \Cref{th:concrete-maxkcsp} --- see \Cref{cor:concrete-maxkcsp-explicit}]\label{th7c}
There is a classical algorithm which, for any instance of unweighted MAX-\(k\)-CSP, produces an optimal solution with high probability in time
\[
O^\ast\left(2^{(1-c^{\mathrm{cl}}_\indB)n}\right)\,,\,\,\textrm{where }\,\,
c^{\mathrm{cl}}_\indB
\ge
0.7213
%0.72134(2^k-s)^2\,
%\frac{1}{2^{2k}k^2D}
%\left(\frac{|H_{\mathrm{min}}|}{m}\right)^2\,.
\cdot\frac{\Delta^2}{2^{2k}k^2D}\,.
\]
\end{theorem}
Since 
$\Delta < 1$, we immediately obtain $c^{\mathrm{cl}}_\indB>c^{\mathrm{q}}_\indB$. In addition to having a larger leading constant, $c^{\mathrm{cl}}_\indB$ can be significantly larger than the value of $c^{\mathrm{q}}_\indB$ since its dependence on $\frac{\Delta}{2^{k}k}$ is only quadratic. 

Besides refuting a super-quadratic quantum advantage, \Cref{th7c} gives novel (quantum-inspired) classical algorithms for constraint satisfaction problems, which we find interesting and even surprising. For concreteness, let us interpret this result in the case of MAX-E3-SAT. \Cref{th7c} shows that for mildly regular instances such that $m^\ast \ge(7/8+\delta)m$ for any constant $\delta\in(0,1/8]$, we can obtain a classical algorithm with running time $2^{(1-c)n}$ for some constant $c>0$. In contrast, achieving such a complexity for arbitrary instances of MAX-E3-SAT is a long-standing open problem; 
besides the case $\delta=1/8$ (corresponding to satisfiable 3-SAT, for which faster algorithms are known, e.g., \cite{Scheder24,Schoning02}) we are not aware of any prior work achieving this for constant $\delta> 0$. Thus, our result also suggests that the case $\delta\approx 0$ might be the hardest case.

\paragraph{Faster quantum algorithms.} 
A further consequence of our technique is that the resulting classical algorithms also admit the standard quadratic quantum speedup based on amplitude amplification. In this sense, the same viewpoint yields algorithmic consequences on both the classical and quantum sides
%: the main text develops the classical exact algorithms, while Appendix~\ref{subsec:brief-quantum} records the corresponding routine quantum speedups via amplitude amplification and Grover search.

More precisely, in \Cref{subsec:brief-quantum} we explain how to use amplitude amplification~\cite{BrassardHMT02} and quantum minimum finding~\cite{Durr1999} to obtain a quadratic speedup over the classical algorithms of \Cref{th6c,th7c}, yielding the following results, which improve upon the quantum algorithms of \cite{DalzellPCB23} since $c^{\mathrm{cl}}_\indA>c^{\mathrm{q}}_\indA$ and $c^{\mathrm{cl}}_\indB>c^{\mathrm{q}}_\indB$.
% a quantum algorithm for MAX-E$k$-LIN2 with running time
% \[    
%     O^\ast\!\left(2^{(1-c^{\mathrm{cl}}_\indA)n/2}\right)
% \]
% when Assumption~(\ref{eq:cond}) holds, and a quantum algorithm for MAX-k-CSPs with running time
% \[
%     O^\ast\!\left(2^{(1-c^{\mathrm{cl}}_\indB)n/2}\right)
% \]
% unconditionally. This improves the quantum algorithms of \cite{DalzellPCB23}.

\begin{theorem}\label{th6q2}
%For any parameters $\gamma\in[(1+4\log_2 n)/n,1]$ and $\eta\in[0,1]$,
%there is a classical algorithm which, for any instance of MAX-E\(k\)-LIN2 with $n\ge 6k$ such that \cref{eq:cond} holds, 
Under the same conditions as in \Cref{th6}, there is a quantum algorithm which, for any instance of MAX-E\(k\)-LIN2 such that Condition \eqref{eq:cond} holds, 
produces an optimal solution with high probability in time
\[
O^\ast\left(2^{(1-c^{\mathrm{cl}}_\indA)n/2}\right)\,.
\]
\end{theorem}
\begin{theorem}\label{th7q2}
There is a quantum algorithm which, for any instance of unweighted MAX-\(k\)-CSP, produces an optimal solution with high probability in time
\[
O^\ast\left(2^{(1-c^{\mathrm{cl}}_\indB)n/2}\right)\,.
\]
\end{theorem}

%wdo get a significant improvement. 
%The classical comparison scale is quadratic in the normalized optimum parameter, whereas the quantum comparison scale in Dalzell et al.'s Theorem~7 is cubic in their unweighted normalization.

\subsection{Overview of our techniques}
We isolate an explicit classical conditioning-and-search framework for exact optimization that captures a substantial part of the mechanism underlying short-path quantum algorithms. We then instantiate this framework in two settings: a correlated-pair analysis for MAX-E\(k\)-LIN2, and a local-Lipschitz analysis for (weighted) MAX-\(k\)-CSP.

\paragraph{Conditioning-and-search.}
The common structure underlying our classical algorithm is the conditioning-and-search framework of Section~\ref{sub:abstract}. The framework is simple: sample from a \emph{conditioning set}~$T$ and then perform exhaustive search in a region associated with the sampled point. Algorithmically, the task therefore splits into two parts:
\begin{itemize}
    \item identify $T$ and upper bound its size, and
    \item identify a sufficiently large subset of sampled points $S\subseteq T$ from which a local exhaustive search recovers an optimum (this subset $S$ is called the \emph{successful set}).
\end{itemize}
In our applications of this abstract framework (from Section \ref{sub:2.2} onward), the conditioning set is the \emph{near-optimal threshold set} defined as follows (for some parameter \(\eta\in(0,1)\)):
\[
T=\{x:H(x)\le (1-\eta)H_{\mathrm{min}}\}\,.
\]
%The construction of the successful sets $S$ used in our applications will be described later.

This leads to a two-parameter description of the running time. 
The first parameter is a \emph{threshold exponent} \(\gamma\), defined via an upper bound of the form
\begin{equation}\label{eq1}
|T|\le 2^{(1-\gamma)n}\,.
\end{equation}
The second parameter is a \emph{successful-set exponent} \(\kappa\), defined via a lower bound of the form
\begin{equation}\label{eq2}
|S|\ge 2^{\kappa n}\,.
\end{equation}
%for an explicit set \(S_\eta\subseteq T_\eta\) such that every \(x\in S_\eta\) comes with a searchable region containing an optimum assignment. 

We package these two conditions into a black-box theorem (\Cref{th:black-box}). Informally, this theorem shows that if both (\ref{eq1}) and~(\ref{eq2}) hold,
and if from every point of \(S\) a local exhaustive search recovers an optimum (i.e., if $S$ is a successful set), then the resulting classical algorithm has complexity $2^{(1-c^{\mathrm{cl}})n}$ with 
\[
c^{\mathrm{cl}}=\min\{\gamma,\kappa\}.
\]

Structurally, this can be summarized as
\[
\boxed{
\text{thin near-optimal region}
\;+\;
\text{large successful subset}
}
\quad\Longrightarrow\quad
\boxed{
c^{\mathrm{cl}}=\min\{\gamma,\kappa\},
}
\]
where
\begin{itemize}
    \item \(\gamma\) measures the thinness of the near-optimal region;
    \item \(\kappa\) measures the size of the successful set.
\end{itemize}

\paragraph{Comparison with the short-path quantum algorithm.}
The quantum framework of Dalzell, Pancotti, Campbell, and Brand\~{a}o~\cite{DalzellPCB23} uses the same near-optimal threshold set, with the same threshold parameter \(\eta\). The upper bound (\ref{eq1}) is used to bound the norm of the Hamiltonian constructed in the short-path quantum algorithm. Instead of introducing an explicit successful set and using the lower bound (\ref{eq2}), the quantum proof uses an overlap-propagation mechanism, encoded by a ``subdepolarizing'' property using another parameter $\alpha$. Combined with Grover search, this gives a quantum algorithm with complexity of the form $2^{(1-c^{\mathrm{q}})n/2}$ for $c^{\mathrm q}=\Phi(\gamma,\alpha)$, where $\Phi$ is an explicit function supplied by the analysis of \cite{DalzellPCB23}.

Structurally, the quantum side may thus be summarized as
\[
\boxed{
\text{thin near-optimal region}
\;+\;
\text{large overlap propagation}
}
\quad\Longrightarrow\quad
\boxed{
c^{\mathrm q}=\Phi(\gamma,\alpha),
}
\]
where
\begin{itemize}
    \item \(\gamma\) again measures the thinness of the near-optimal region;
    \item \(\alpha\) is the overlap-propagation parameter.
\end{itemize}

Our analysis (outlined below) shows that for MAX-E\(k\)-LIN2 and MAX-\(k\)-CSP we have $c^{\mathrm{cl}}>c^{\mathrm q}$, which implies that the short-path quantum algorithm does not achieve a super-quadratic advantage. 

\paragraph{Constructing the successful sets.}
\Cref{sec:favorable-sets} (and \Cref{sec:technical-details} which contains some technical details omitted in \Cref{sec:favorable-sets}) are devoted to constructing the successful sets \(S\) in two different settings.

\begin{itemize}
\item
In the first setting, relevant to MAX-E\(k\)-LIN2 and the comparison with \Cref{th6}, the set $T$ is thin (i.e., the upper bound (\ref{eq1}) holds) by assumption.
%the thin-sublevel-set condition is assumed as part of the input hypothesis. 
The task is therefore only to identify a large successful subset \(S\subseteq T\). Section~\ref{subsec:noise-stability-analysis} does this via a correlated-pair analysis on the Boolean cube. Starting from an optimum assignment \(x^\ast\), we generate a random point~\(X\) by independently flipping each bit of \(x^\ast\) with a suitable probability; equivalently, \(x^\ast\) and \(X\) form a \(\rho\)-correlated pair on \(\{-1,1\}^n\), for some parameter $\rho$. We then use the exact identity
\[
\mathbb E[H(X)]=\rho^k H_{\mathrm{min}}
\]
for homogeneous degree-\(k\) objectives. This gives an inverse-polynomial lower bound on the probability that~\(X\) falls in the target threshold set. Restricting further to a narrow Hamming layer around the expected distance from \(x^\ast\) converts that probability bound into a counting statement for an explicit successful set \(S\). The resulting search radius (which controls the locality, and thus the cost, of searching in the region associated with a point of $S$) is then governed by the entropy of that Hamming layer.
\item
In the second setting, relevant to MAX-\(k\)-CSP and the comparison with \Cref{th7}, the thinness of \(T\) is not assumed in advance but instead derived from the instance. More precisely, for a chosen \(\eta\), a concentration argument based on McDiarmid's inequality gives an explicit threshold exponent in terms of the intrinsic parameters of the instance, namely the irregularity parameter \(D\)
and the normalized optimum scale $\Delta$.
%the quantity \(\Sigma\)\francois{replace by "a quantity $\Sigma$ that represents..."? Or simply not introduce $\Sigma$ in the intro and use $\Delta$ instead of \(|H_{\mathrm{min}}|/\Sigma\)?}, and the ratio \(|H_{\mathrm{min}}|/\Sigma\). 
Section~\ref{subsec:local-lipschitz-analysis} then supplies the corresponding successful set by a local-Lipschitz argument: flipping a light variable changes the objective \(H\) by only a controlled amount, so a restricted Hamming ball around an optimum assignment remains inside the target threshold set. This yields an explicit successful set \(S\subseteq T\), and the resulting exponent is controlled by the entropy of that restricted ball.
% In the second setting, relevant to weighted MAX-\(k\)-CSPs and the comparison with \Cref{th7}, the thinness of \(T\) is not assumed in advance but instead derived from the instance. More precisely, for a chosen \(\eta\), a concentration argument based on McDiarmid's inequality gives an exponent
% \[
% \gamma=\Gamma\!\left(k,D,\Delta,\eta\right).
% \]
% Section~\ref{subsec:local-lipschitz-analysis} then supplies the corresponding successful set by a local-Lipschitz argument: flipping a light variable changes the objective \(H\) by only a controlled amount, so a restricted Hamming ball around an optimum assignment remains inside the target threshold set. This yields an explicit successful set \(S\subseteq T\), and the resulting exponent is controlled by the entropy of that restricted ball.
\end{itemize}

\paragraph{Combining all elements.}
Section~\ref{sec:concrete-corollaries} then combines these successful-set constructions with the corresponding threshold-set bounds. In this way, the paper separates the argument into three layers: an abstract conditioning-and-search theorem, two constructions of explicit successful sets, and concrete corollaries obtained by combining those constructions with upper bounds on the threshold set. From this perspective, the contrast with the quantum side is also clearer: the classical analysis proceeds by exhibiting and counting explicit successful sets, whereas the quantum analysis is naturally phrased in terms of amplitudes and overlaps.

\subsection{Related work and organization}

\paragraph{Other related works.}
The short-path framework was developed further by Chakrabarti et al.~\cite{Chakrabarti-etal24}. We briefly discuss this work in \Cref{sec:concluding}.

This paper also fits into the literature on dequantization: Tang's breakthrough result on recommendation systems~\cite{Tang19,Tang23}, subsequent low-rank and QSVT dequantization works \cite{BakshiT24, GharibianLG23,ChiaGLLTW20,ChiaGLLTW22,GilyenST22,KothariOW26,LeGall25}, the recent classical quadratic-speedup result of Gupta, He, O'Donnell and Singer~\cite{GuptaHOS26} for planted $k$XOR, and the result of Kothari, O'Donnell and Wu on \(\mathrm{SIS}^{\infty}\)~\cite{KothariOW26}, all illustrate that some quantum speedups can admit classical counterparts once the underlying mechanism is made explicit in classical terms. In contrast to the mainly linear-algebraic dequantization frameworks in that literature, our focus here is on extracting the structure behind the analysis of the short-path optimization in \cite{DalzellPCB23,Hastings18short,Hastings18weaker,Hastings19duality} and reformulating it as a purely classical conditioning-and-search mechanism.

Finally, we mention Decoded Quantum Interferometry (DQI), a recent framework for quantum optimization introduced by Jordan et al.~\cite{JordanSWZSKIKB25} that combines phase encoding and interference with a classical decoding step. The central idea is to map structured optimization problems to decoding problems via the quantum Fourier transform, thereby enabling the use of efficient classical decoding algorithms to extract solutions. In particular, the original work identifies sparse MAX-$k$-XORSAT and MAX-LINSAT as natural applications of this approach, where the constraint structure can be interpreted as a code and exploited by decoding procedures. Subsequent work has clarified both the scope and the limitations of this framework. On the one hand, DQI can yield strong performance guarantees on instances with suitable algebraic or coding-theoretic structure. %On the other hand, several results indicate that such advantages are inherently structure-dependent: for example, for problems such as MaxCut or random MAX-$k$-XOR-SAT, nontrivial guarantees appear to be restricted to instances that are already classically tractable or that avoid overlap-gap phenomena.
Overall, the current understanding suggests that DQI is best viewed not as a general-purpose algorithm for arbitrary MAX-$k$-CSPs, but rather as a promising approach for specially structured instances where efficient decoding is possible \cite{JordanSWZSKIKB25,AnschuetzGamarnikLuDQI,KramerSchubertEisertDQI}. 

\paragraph{Organization of the paper.}
Section~\ref{sec:framework} introduces the abstract two-parameter conditioning-and-search framework and the associated black-box theorem. Section~\ref{sec:favorable-sets} constructs successful sets in the two regimes considered in the paper: the correlated-pair regime for homogeneous degree-\(k\) objectives and the local-Lipschitz regime for MAX-\(k\)-CSP. Section~\ref{sec:concrete-corollaries} combines these successful-set constructions with the corresponding threshold exponents to derive the main concrete exact-optimization corollaries. Section~\ref{sec:technical-details} collects the quantitative estimates and standard calculations deferred from the main text. Section~\ref{sec:concluding} contains concluding remarks.

The appendices collect supplementary material. %Appendix~\ref{sec:structural-comparison} compare the classical and the quantum approaches. 
Appendices \ref{sec:standard-variants}, \ref{sec:top-k-variant} and \ref{sec:optimistic-first-local-ball}
discuss the case where the optimal value is unknown. %Appendix~A.2 records the corresponding routine quadratic quantum speedups via amplitude amplification and Grover search.
Appendix \ref{sec:notation-guide} collects the main notation used in Sections 2, 3, and 4.

% ==================================================
% Section 2: A two-parameter framework for exact optimization
% ==================================================
\section{A Two-Parameter Framework for Exact Optimization}\label{sec:framework}
The role of this section is to isolate our framework for exact optimization abstractly before turning to two concrete instantiations in later sections.
Throughout this section, \(n\) denotes the size parameter, and we assume that the input instance has encoding length \(\mathrm{poly}(n)\). Accordingly, all polynomial-time statements are with respect to $n$. Here and throughout this section, a set \(T\subseteq \Omega_n\) is called polynomial-time testable if, given \(x\in\Omega_n\), one can decide whether \(x\in T\) in deterministic polynomial time.

% ==================================================
\subsection{Abstract conditioning-and-search theorem}\label{sub:abstract}

The first ingredient is an abstract theorem showing that, once one can sample from a conditioning set $T$ and search a region associated with each sample, the expected running time is governed by the probability of landing in a successful part of that set.

\begin{theorem}[Abstract conditioning-and-search principle]\label{th:abstract-conditioning-search}
Let \(\Omega_n\) be a finite search space, let \(T\subseteq\Omega_n\) be a polynomial-time testable set, let \(\mu\) be a probability distribution supported on \(T\), and let \(f\colon\Omega_n\to\{0,1\}\) be a polynomial-time computable target predicate. Suppose that:
\begin{enumerate}
    \item there is a randomized sampler such that each invocation, using fresh independent randomness, outputs a sample \(x\sim \mu\) and has expected running time \(\tau_{\mathrm{samp}}(n)\);
    \item for each \(x\in T\), there is a region \(\mathcal N(x)\subseteq\Omega_n\) that can be exhaustively searched in time at most \(\tau_{\mathrm{search}}(n)\);
    \item there exist a set \(S\subseteq T\) and a target point \(x^\ast\in\Omega_n\) with \(f(x^\ast)=1\) such that
    \[
    x^\ast\in \mathcal N(x)
    \qquad\text{for every }x\in S;
    \]
    \item %the successful mass $\mu(S)$ satisfies
    $
    \mu(S)\ge \beta_n
    $
    for some parameter \(\beta_n>0\).
\end{enumerate}
Then the algorithm that repeatedly and independently samples \(x\sim \mu\) and searches \(\mathcal N(x)\) until it finds some \(y\) with \(f(y)=1\) has expected running time at most
\[
O\!\left(\frac{\tau_{\mathrm{samp}}(n)+\tau_{\mathrm{search}}(n)}{\beta_n}\right).
\]
\end{theorem}

%%%%%%%%%%%%%%%%%%%%%%%%%%%%%%%%%%%%%%%%%%%%%%%%%%
%%%%%%%%%%%%%%%%%%%%%%%%%%%%%%%%%%%%%%%%%%%%%%%%%%
%%%%%%%%%%%%%%%%%%%%%%%%%%%%%%%%%%%%%%%%%%%%%%%%%%
\begin{proof}
Let \(N\) be the index of the first successful iteration, where an iteration is called successful if the sampled point lies in \(S\). By Assumption 3, every successful iteration finds a point \(y\) with \(f(y)=1\), since then \(x^\ast\in\mathcal N(x)\).

For each \(i\ge 1\), condition on the random choices made in the first \(i-1\) iterations. On the event \(\{N\ge i\}\), the \(i\)-th iteration is executed. Because that iteration uses fresh independent randomness and outputs a sample distributed as \(\mu\), we have
\[
\Pr[N=i\mid \text{first \(i-1\) iterations}]
\ge \beta_n
\qquad\text{on the event }\{N\ge i\}.
\]
Equivalently,
\[
\Pr[N\ge i+1\mid \text{first \(i-1\) iterations}]
\le 1-\beta_n
\qquad\text{on the event }\{N\ge i\}.
\]
Taking expectations and iterating gives
\[
\Pr[N\ge i]\le (1-\beta_n)^{i-1}
\qquad (i\ge 1).
\]
Hence
\[
\mathbb E[N]
=
\sum_{i\ge 1}\Pr[N\ge i]
\le
\sum_{i\ge 1}(1-\beta_n)^{i-1}
=
\frac{1}{\beta_n}.
\]

Now let \(C_i\) be the running time of the \(i\)-th iteration. Since the \(i\)-th invocation of the sampler has expected cost at most \(\tau_{\mathrm{samp}}(n)\), and the subsequent exhaustive search costs at most \(\tau_{\mathrm{search}}(n)\), we have
\[
\mathbb E[C_i\mid \text{first \(i-1\) iterations}]
\le
\tau_{\mathrm{samp}}(n)+\tau_{\mathrm{search}}(n)
\qquad (i\ge 1).
\]
Therefore
\[
\mathbb E\!\left[\sum_{i=1}^{N} C_i\right]
=
\sum_{i\ge 1}\mathbb E\!\left[\mathbf 1[N\ge i]\,C_i\right].
\]
Since the event \(\{N\ge i\}\) depends only on the outcomes of the first \(i-1\) iterations,
\begin{align*}
\mathbb E\!\left[\mathbf 1[N\ge i]\,C_i\right]
&=
\mathbb E\!\left[\mathbf 1[N\ge i]\,
\mathbb E[C_i\mid \text{first \(i-1\) iterations}]\right]\\
&
\le
\bigl(\tau_{\mathrm{samp}}(n)+\tau_{\mathrm{search}}(n)\bigr)\Pr[N\ge i].
\end{align*}
Summing over \(i\ge 1\), we obtain
\begin{align*}
\mathbb E\!\left[\sum_{i=1}^{N} C_i\right]
&\le
\bigl(\tau_{\mathrm{samp}}(n)+\tau_{\mathrm{search}}(n)\bigr)
\sum_{i\ge 1}\Pr[N\ge i]\\
&=
\bigl(\tau_{\mathrm{samp}}(n)+\tau_{\mathrm{search}}(n)\bigr)\mathbb E[N]\\
&\le
\frac{\tau_{\mathrm{samp}}(n)+\tau_{\mathrm{search}}(n)}{\beta_n}\,,
\end{align*}
as desired.
\end{proof}
%%%%%%%%%%%%%%%%%%%%%%%%%%%%%%%%%%%%%%%%%%%%%%%%%%
%%%%%%%%%%%%%%%%%%%%%%%%%%%%%%%%%%%%%%%%%%%%%%%%%%
%%%%%%%%%%%%%%%%%%%%%%%%%%%%%%%%%%%%%%%%%%%%%%%%%%

The uniform case is the main specialization used in this paper.

\begin{corollary}[Uniform specialization]\label{cor:abstract-conditioning-search-uniform}
In the setting of \Cref{th:abstract-conditioning-search}, if \(\mu\) is the uniform distribution on \(T\), then the expected running time is at most
\[
O\!\left(\frac{|T|}{|S|}\bigl(\tau_{\mathrm{samp}}(n)+\tau_{\mathrm{search}}(n)\bigr)\right).
\]
\end{corollary}
\begin{proof}
For uniform \(\mu\), one has \(\mu(S)=|S|/|T|\). The statement then follows by using \Cref{th:abstract-conditioning-search} with \(\beta_n=|S|/|T|\).
\end{proof}

We will typically obtain uniform samples from \(T\) by rejection sampling from an ambient space~\(\Omega_n\): repeatedly draw a uniformly random point of \(\Omega_n\) until the sampled point lies in \(T\), and then output that accepted point. Because all accepted points arise from the same uniform distribution on \(\Omega_n\), the output is then uniform on \(T\). This leads to the following corollary.

\begin{corollary}[Rejection-sampling form]\label{cor:abstract-conditioning-search-rejection}
In the setting of \Cref{th:abstract-conditioning-search}, suppose in addition that \(\mu\) is the uniform distribution on \(T\) and that it is sampled by rejection sampling from a finite ambient space \(\Omega_n\), with each raw sample and membership test costing \(\mathrm{poly}(n)\) time. Then the expected running time is
\[
O\!\left(\frac{|\Omega_n|}{|S|}\,\mathrm{poly}(n)+\frac{|T|}{|S|}\cdot\tau_{\mathrm{search}}(n)\right).
\]
\end{corollary}
\begin{proof}
The rejection sampler accepts with probability \(|T|/|\Omega_n|\). Hence its expected cost is
\[
O\!\left(\frac{|\Omega_n|}{|T|}\,\mathrm{poly}(n)\right).
\]
Substituting this bound into \Cref{cor:abstract-conditioning-search-uniform} yields the claim.
\end{proof}

The main point of these statements is that exact optimization reduces to two quantitative inputs: an upper bound on the conditioning set $T$ and a lower bound on the size of the successful subset $S$ from which exhaustive search provably reaches an optimum.

% ==================================================
\subsection[A black-box theorem from $(\gamma_\eta,\kappa_\eta)$]{A black-box theorem from \texorpdfstring{$\boldsymbol{(\gamma_\eta,\kappa_\eta)}$}{(gamma, kappa)}}\label{sub:2.2}

We now package the abstract theorem in the form used throughout the rest of the paper. In this subsection we specialize to the ambient space
\[
\Omega_n=\{-1,1\}^n.
\]

Fix \(\eta\in(0,1)\), and define the threshold set
\[
T_\eta=\{x:H(x)\le (1-\eta)H_{\mathrm{min}}\}.
\]
Suppose that:
\begin{enumerate}
    \item \(T_\eta\) is polynomial-time testable;
    \item there exists an explicit set \(S_\eta\subseteq T_\eta\);
    \item for each \(x\in T_\eta\), there is a search region \(\mathcal N_\eta(x)\subseteq\Omega_n\);
    \item for every \(x\in S_\eta\), the region \(\mathcal N_\eta(x)\) contains an optimum assignment;
    \item the threshold set is exponentially thin:
    \[
    |T_\eta|\le 2^{(1-\gamma_\eta)n};
    \]
    \item the successful set is exponentially large:
    \[
    |S_\eta|\ge 2^{\kappa_\eta n-o(n)};
    \]
    \item the region \(\mathcal N_\eta(x)\) can be exhaustively searched in time
    \[
    2^{\kappa_\eta n+o(n)}
    \]
    for every \(x\in T_\eta\).
\end{enumerate}

Then the abstract conditioning-and-search theorem immediately yields the following statement.

\begin{theorem}[Black-box exact optimization from threshold and successful-set bounds]\label{th:black-box}
Under the seven assumptions above, an optimum assignment can be found in expected time
\[
2^{(1-c_\eta)n+o(n)},
\qquad
c_\eta=\min\{\gamma_\eta,\kappa_\eta\}.
\]
\end{theorem}
\begin{proof}
Apply \Cref{cor:abstract-conditioning-search-rejection} with ambient space \(\Omega_n=\{-1,1\}^n\), conditioning set \(T=T_\eta\), successful set \(S=S_\eta\), and neighborhood family \(\mathcal N(x)=\mathcal N_\eta(x)\). By Assumption 7, for every \(x\in T_\eta\), the region \(\mathcal N_\eta(x)\) can be exhaustively searched in time \(2^{\kappa_\eta n+o(n)}\), and by Assumption 4, for every \(x\in S_\eta\) it contains an optimum assignment. Hence the expected running time is
\[
O\!\left(\frac{2^n}{|S_\eta|}\,\mathrm{poly}(n)+\frac{|T_\eta|}{|S_\eta|}\,2^{\kappa_\eta n+o(n)}\right).
\]
Using
\[
|S_\eta|\ge 2^{\kappa_\eta n-o(n)}
\qquad\text{and}\qquad
|T_\eta|\le 2^{(1-\gamma_\eta)n},
\]
we obtain
\[
\mathbb E[\mathrm{time}]
\le
2^{(1-\kappa_\eta)n+o(n)}+2^{(1-\gamma_\eta)n+o(n)}
=
2^{(1-\min\{\gamma_\eta,\kappa_\eta\})n+o(n)}\,,
\]
as claimed.
\end{proof}

The point of \Cref{th:black-box} is that, once \(\gamma_\eta\) and \(\kappa_\eta\) are available, the running-time exponent becomes completely explicit.

% We will use the following notation for the two exponents constructed later:
% \[
% \kappa_\eta=
% \begin{cases}
% \kappa_\eta^{\mathrm{ns}} & \text{in the correlated-pair regime,}\\[1mm]
% \kappa_\eta^{\mathrm{lip}} & \text{in the local-Lipschitz regime.}
% \end{cases}
% \]

% ==================================================
\subsection{Two ways threshold exponents arise}\label{subsec:threshold-exponent-cases}\label{sub:2.3}

The black-box theorem above takes as input a threshold parameter \(\eta\in(0,1)\) and a threshold exponent \(\gamma_\eta\) such that
\[
|T_\eta|\le 2^{(1-\gamma_\eta)n},
\qquad
T_\eta=\{x:H(x)\le (1-\eta)H_{\mathrm{min}}\}.
\]
There are two conceptually different ways in which such a pair \((\gamma_\eta,\eta)\) may arise.

\paragraph{Case 1: \(\boldsymbol{(\gamma,\eta)}\) given beforehand.}
In some settings (e.g., the proof of \Cref{th6c}), one assumes from the outset that, for a fixed \(\eta\in(0,1)\), the set \(T_\eta\) satisfies
\[
|T_\eta|\le 2^{(1-\gamma)n}
\]
for some given exponent \(\gamma>0\).
In this regime, the role of the analysis is only to construct a successful set \(S_\eta\subseteq T_\eta\) and to compute the corresponding exponent \(\kappa_\eta\). Once this is done, \Cref{th:black-box} yields a classical exponent of the form
\[
\min\{\gamma,\kappa_\eta\}.
\]

\paragraph{Case 2: instance parameters determine \(\boldsymbol{(\gamma_\eta,\eta)}\).}
In other settings (e.g., the proof of \Cref{th7c}), one does not assume a threshold exponent as part of the input. Instead, one derives a bound of the form
\[
|T_\eta|\le 2^{(1-\gamma_\eta)n}
\]
from structural parameters of the instance. In this case, the pair \((\gamma_\eta,\eta)\) is produced internally by the analysis.

In the MAX-\(k\)-CSP setting considered later, the threshold exponent is derived from the intrinsic parameters of the instance, namely the irregularity parameter \(D\) and the normalized optimum scale~$\Delta$ (in the case of weighted MAX-\(k\)-CSP, $\Delta$ is replaced by the ratio \(|H_{\mathrm{min}}|/\Sigma\), where $\Sigma$ is defined in \cref{subsec:local-lipschitz-analysis}).
%the quantity \(\Sigma\)\francois{same here}, and the ratio \(|H_{\mathrm{min}}|/\Sigma\). 
A concentration argument based on McDiarmid's inequality then yields an explicit threshold exponent \(\gamma_\eta\). Once such a threshold exponent has been derived, \Cref{th:black-box} again yields a classical algorithm with exponent of the form
\[
\min\{\gamma_\eta,\kappa_\eta\}.
\]
% In the weighted MAX-\(k\)-CSP setting considered later, the relevant parameters include the irregularity parameter \(D\), the quantity \(\Sigma\), and the normalized optimum scale \(|H_{\mathrm{min}}|/\Sigma\). A concentration argument, based on McDiarmid's inequality, then yields a threshold exponent
% \[
% \gamma_\eta=\Gamma\!\left(D,\frac{|H_{\mathrm{min}}|}{\Sigma},\eta\right)
% \]
% for an explicit function \(\Gamma\).
% In the weighted MAX-\(k\)-CSP setting considered later, the relevant parameters include the arity~\(k\), the irregularity parameter \(D\), and the normalized optimum scale \(|H_{\mathrm{min}}|/W\). A concentration argument, typically based on McDiarmid's inequality, then yields a threshold exponent
% \[
% \gamma_\eta=\Gamma\!\left(k,D,\frac{|H_{\mathrm{min}}|}{W},\eta\right)
% \]
% for an explicit function \(\Gamma\). 
% Once such a threshold exponent has been derived, \Cref{th:black-box} again yields a classical algorithm with exponent of the form
% \[
% \min\{\gamma_\eta,\kappa_\eta\}.
% \]

\paragraph{Outlook.}
Thus the rest of the paper fits into the following common pattern:
\[
\boxed{
\text{either receive }(\gamma,\eta) \text{ beforehand or derive }
(\gamma_\eta,\eta)\text{ from instance parameters}
}
\]
followed by
\[
\boxed{
\text{construct a successful set }S_\eta
\text{ and compute }\kappa_\eta.
}
\]
The resulting classical exponent is then obtained from \Cref{th:black-box} as
\[
\min\{\gamma,\kappa_\eta\}
\qquad\text{or}\qquad
\min\{\gamma_\eta,\kappa_\eta\},
\]
depending on which of the two cases applies.

The next section constructs the sets \(S_\eta\) needed in these two regimes. 

% ==================================================
% Section 3: Analysis of successful sets
% ==================================================
\section{Analysis of Successful Sets}\label{sec:favorable-sets}

This section provides the exponent \(\kappa_\eta\) that appears in \Cref{th:black-box}. Standard quantitative estimates used in the derivation of the final concrete corollaries are deferred to Section~\ref{sec:technical-details}.

We do this in two different ways:

\begin{itemize}
\item
In the first construction, corresponding to Case~1 of Section~\ref{subsec:threshold-exponent-cases}, a threshold exponent \(\gamma\) is given as part of the hypothesis, and the task is only to identify a large subset \(S_\eta\subseteq T_\eta\) from which local exhaustive search recovers an optimum. This leads to a correlated-pair analysis for MAX-E$k$-LIN2, in which the successful set is a narrow Hamming layer around an optimum assignment. The resulting exponent is denoted by \(\kappa_\eta^{\mathrm{ns}}\) and will be expressed later in terms of a suitable flip-rate parameter.

\item
In the second construction, corresponding to Case~2 of Section~\ref{subsec:threshold-exponent-cases}, the threshold exponent \(\gamma_\eta\) is derived from intrinsic parameters of the instance.
%, so dependence on 
%the ratio \(|H_{\mathrm{min}}|/\Sigma\)\francois{mention that this ratio is equal to the quantity $\Delta$ used in the intro?} is already present on the threshold side. 
This leads to a local-Lipschitz analysis for weighted MAX-\(k\)-CSP, in which the successful set is a restricted Hamming ball around an optimum assignment. The resulting exponent is denoted by \(\kappa_\eta^{\mathrm{lip}}\) and will be expressed later in terms of the corresponding local-ball radius parameter.
% In the second construction, corresponding to Case~2 of Section~\ref{subsec:threshold-exponent-cases}, the threshold exponent \(\gamma_\eta\) is derived from instance parameters, so dependence on the normalized optimum scale (the parameter \(\Delta\) introduced in the introduction) is already present on the threshold side. This leads to a local-Lipschitz analysis for weighted MAX-\(k\)-CSPs, in which the successful set is a restricted Hamming ball around an optimum assignment. The resulting exponent is denoted by \(\kappa_\eta^{\mathrm{lip}}\) and will be expressed later in terms of the corresponding local-ball radius parameter.
\end{itemize}

% We do this in two different ways:

% \begin{itemize}
% \item
% In the first construction, corresponding to Case~1 of Section~\ref{subsec:threshold-exponent-cases}, a threshold exponent \(\gamma\) is given as part of the hypothesis, and the task is only to identify a large subset \(S_\eta\subseteq T_\eta\) from which local exhaustive search recovers an optimum. This leads to a correlated-pair analysis for homogeneous degree-\(k\) objectives, in which the successful set is a typical Hamming shell around an optimum assignment and
% \[
% \kappa_\eta^{\mathrm{ns}}=h(q_\eta).
% \]

% \item
% In the second construction, corresponding to Case~2 of Section~\ref{subsec:threshold-exponent-cases}, the threshold exponent \(\gamma_\eta\) is derived from instance parameters, so dependence on the normalized optimum scale (the parameter $\Delta$ introduced in the introduction) is already present on the threshold side. This leads to a local-Lipschitz analysis for weighted MAX-\(k\)-CSPs, in which the successful set is a restricted Hamming ball around an optimum assignment and
% \[
% \kappa_\eta^{\mathrm{lip}}=\frac12\,h(\theta_\eta).
% \]
% \end{itemize}

%The purpose of the present section is thus to produce the exponent \(\kappa_\eta\) in the two regimes used later. Standard shell-counting, binomial-entropy, and concentration estimates are deferred to Section~\ref{sec:technical-details}.

Remember that $h(\cdot)$ denotes the binary entropy function:
\[
h(t)=-t\log_2 t-(1-t)\log_2(1-t)
\qquad (0\le t\le 1)\,.
\]
%for the binary entropy function.

% ==================================================
\subsection{Correlated-pair analysis}\label{subsec:noise-stability-analysis}
We begin with the homogeneous degree-\(k\) setting, which is the formal framework underlying MAX-E\(k\)-LIN2 in this paper. Thus, in this subsection, we consider objective functions of the form
\[
H(x)=\sum_{S\in\mathcal F} c_S\prod_{i\in S}x_i,
\qquad x\in\{-1,1\}^n,
\]
where \(\mathcal F\) is a collection of \(k\)-element subsets of \([n]\), and \(c_S\in\mathbb R\) for each $S\in \mathcal F$. In particular, every monomial has degree exactly \(k\).

% We begin with the homogeneous degree-\(k\) setting. This is the setting relevant to Case~1 of Section~\ref{subsec:threshold-exponent-cases}, where a threshold exponent \(\gamma\) is assumed as input and the analysis only needs to produce a large successful subset of the near-optimal region.

% Throughout this subsection, \(k\ge 2\) is fixed, and
% \[
% H(x)=\sum_{S\in\mathcal F} c_S\prod_{i\in S}x_i,
% \qquad x\in\{-1,1\}^n,
% \]
% is a multilinear polynomial on the Boolean cube in which every monomial has degree exactly \(k\).

If \(H\equiv 0\), then every assignment is optimal and the discussion below is vacuous. Thus we assume from now on that \(H\not\equiv 0\). Since every monomial has positive degree, the average of \(H\) over the discrete cube is zero:
\[
\mathbb E_{x\in\{-1,1\}^n}[H(x)]=0.
\]
Because \(H\) is not identically zero, it takes both positive and negative values. In particular,
\[
H_{\mathrm{min}}=\min_{x\in\{-1,1\}^n}H(x)<0.
\]
Fix an optimum assignment
\[
x^\ast\in\{-1,1\}^n
\qquad\text{such that}\qquad
H(x^\ast)=H_{\mathrm{min}}.
\]

For \(\eta\in(0,1)\), define the threshold set
\[
T_\eta=\{x\in\{-1,1\}^n:H(x)\le (1-\eta)H_{\mathrm{min}}\}.
\]

We use the standard language of \(\rho\)-correlated pairs on the Boolean cube; see, e.g., Section~2.4 of \cite{ODonnell14}. In the present pure-degree setting, the only fact we need is that for every degree-\(k\) monomial, taking expectation under \(\rho\)-correlated noise multiplies that monomial by \(\rho^k\).

\begin{proposition}[Correlated-pair contraction]\label{prop:ns-contraction}
Let \(t=(t_1,\dots,t_n)\in\{-1,1\}^n\) be a random vector whose coordinates are independent and satisfy
\[
\Pr[t_i=-1]=q,
\qquad
\Pr[t_i=1]=1-q,
\]
and write
\[
\rho=\mathbb E[t_i]=1-2q.
\]
Let \(\odot\) denote coordinatewise multiplication, i.e.,
\[
(x\odot y)_i=x_i y_i \qquad (i\in[n]).
\]
Define the random variable
\[
X=x^\ast\odot t.
\]
Then each coordinate \(X_i\) satisfies
\[
\Pr[X_i=x_i^\ast]=1-q,
\qquad
\Pr[X_i=-x_i^\ast]=q,
\]
i.e., \(X\) is obtained from \(x^\ast\) by independently flipping each coordinate with probability \(q\) (equivalently, the pair \((x^\ast,X)\) is a \(\rho\)-correlated pair on \(\{-1,1\}^n\)). Then
\[
\mathbb E[H(X)]=\rho^k H_{\mathrm{min}}.
\]
\end{proposition}
% \begin{proposition}[Correlated-pair contraction]\label{prop:ns-contraction}
% Let \(t=(t_1,\dots,t_n)\in\{-1,1\}^n\) be a random vector whose coordinates are independent and satisfy
% \[
% \Pr[t_i=-1]=q,
% \qquad
% \Pr[t_i=1]=1-q,
% \]
% and write
% \[
% \rho=\mathbb E[t_i]=1-2q.
% \]
% Define
% \[
% X=x^\ast\odot t.
% \]
% Equivalently, \((x^\ast,X)\) is a \(\rho\)-correlated pair \francois{The meaning of ``correlated pair'' has not been defined (I think many readers will not know the meaning).} on \(\{-1,1\}^n\). Then
% \[
% \mathbb E[H(X)]=\rho^k H_{\mathrm{min}}.
% \]
% \end{proposition}

\begin{proof}
Write
\[
H(x^\ast\odot t)=\sum_{S\in\mathcal F} b_S\prod_{i\in S} t_i,
\qquad
b_S=c_S\prod_{i\in S}x_i^\ast.
\]
Since every monomial has degree exactly \(k\), independence gives
\[
\mathbb E\!\left[\prod_{i\in S} t_i\right]=\rho^k
\qquad (S\in\mathcal F).
\]
Hence
\[
\mathbb E[H(X)]
=
\sum_{S\in\mathcal F} b_S\,\rho^k
=
\rho^k\sum_{S\in\mathcal F} b_S
=
\rho^k H(x^\ast)
=
\rho^k H_{\mathrm{min}},
\]
as claimed.
\end{proof}

This immediately gives a one-sided bound for the event that \(X\) lands in the threshold set \(T_\eta\).

\begin{proposition}[One-sided lower-tail bound]\label{prop:ns-tail}
In the setting of \Cref{prop:ns-contraction},
\[
\Pr\bigl[H(X)\le (1-\eta)H_{\mathrm{min}}\bigr]
\ge
1-\frac{1-\rho^k}{\eta}.
\]
\end{proposition}

\begin{proof}
Define
\[
Y=H(X)-H_{\mathrm{min}}.
\]
Since \(x^\ast\) is optimal, \(Y\ge 0\). Moreover, by \Cref{prop:ns-contraction},
\[
\mathbb E[Y]
=
\mathbb E[H(X)]-H_{\mathrm{min}}
=
(\rho^k-1)H_{\mathrm{min}}
=
(1-\rho^k)|H_{\mathrm{min}}|.
\]
Finally,
\[
H(X)>(1-\eta)H_{\mathrm{min}}
\iff
Y>\eta |H_{\mathrm{min}}|.
\]
Markov's inequality therefore gives
\[
\Pr\bigl[H(X)>(1-\eta)H_{\mathrm{min}}\bigr]
\le
\frac{\mathbb E[Y]}{\eta |H_{\mathrm{min}}|}
=
\frac{1-\rho^k}{\eta}.
\]
Taking complements proves the claim.
\end{proof}

To convert this probability lower bound into a counting statement, we choose the correlation so that the low-threshold event has inverse-polynomial probability and then restrict to a narrow Hamming layer around the expected number of flipped coordinates.
%To convert this probability lower bound into a counting statement, we choose the correlation so that the low-threshold event has polynomial probability and then restrict to a typical Hamming shell. 

Define
\[
q_{\eta,n}=\frac{1-\left(1-\eta+\frac{\eta}{n}\right)^{1/k}}{2}.
\]
We also define the limiting flip rate
\[
q_\eta=\frac{1-(1-\eta)^{1/k}}{2},
\]
so that
\[
q_{\eta,n}=q_\eta+O_{k,\eta}(1/n).
\]
% Define \(\rho_{\eta,n}\in(0,1)\) and \(q_{\eta,n}\in(0,1/2)\) by
% \[
% \rho_{\eta,n}^k=1-\eta+\frac{\eta}{n},
% \qquad
% q_{\eta,n}=\frac{1-\rho_{\eta,n}}{2}.
% \]
% We also define the limiting flip rate
% \[
% q_\eta=\frac{1-(1-\eta)^{1/k}}{2},
% \]
% so that
% \[
% q_{\eta,n}=q_\eta+O_{k,\eta}(1/n).
% \]
Moreover,
\[
q_\eta\ge \frac{\eta}{2k}.
\]
Indeed, Bernoulli's inequality gives
\[
(1-\eta)^{1/k}\le 1-\frac{\eta}{k},
\]
and therefore
\[
q_\eta=\frac{1-(1-\eta)^{1/k}}{2}\ge \frac{\eta}{2k}.
\]

For \(t\in\{-1,1\}^n\), let
\[
\wt(t)=\bigl|\{i\in[n]:t_i=-1\}\bigr|.
\]
Let \(d_H\) denote the Hamming distance on \(\{-1,1\}^n\). Define the typical shell
\[
A_\eta=
\Bigl\{t\in\{-1,1\}^n:
\bigl|\wt(t)-q_{\eta,n}n\bigr|\le n^{2/3}
\Bigr\},
\]
and the corresponding successful set
\[
S_\eta^{\mathrm{ns}}
=
\Bigl\{
x^\ast\odot t:
t\in A_\eta,\,
H(x^\ast\odot t)\le (1-\eta)H_{\mathrm{min}}
\Bigr\}.
\]
By construction,
\[
S_\eta^{\mathrm{ns}}\subseteq T_\eta.
\]

We will recover an optimum from each successful point by searching a Hamming ball. Define
\[
r_\eta^{\mathrm{ns}}
=
\left\lceil q_\eta n+2n^{2/3}\right\rceil.
\]
Here is the main result of this subsection:
\begin{theorem}[Successful-set bound in the correlated-pair regime]\label{th:ns-favorable}
Fix \(\eta\in(0,1)\). Then, for all sufficiently large \(n\), the set \(S_\eta^{\mathrm{ns}}\) satisfies the following.
\begin{enumerate}
    \item \(S_\eta^{\mathrm{ns}}\subseteq T_\eta\).
    \item For every \(x\in S_\eta^{\mathrm{ns}}\), one has
    \[
    d_H(x,x^\ast)\le r_\eta^{\mathrm{ns}}.
    \]
    In particular, the Hamming ball \(\Ball{x}{r_\eta^{\mathrm{ns}}}\) contains an optimum assignment.
    \item One has
    \[
    |S_\eta^{\mathrm{ns}}|\ge 2^{\kappa_\eta^{\mathrm{ns}}n-o(n)},
    \qquad
    \kappa_\eta^{\mathrm{ns}}=h(q_\eta).
    \]
    \item For every \(x\in\{-1,1\}^n\),
    \[
    |\Ball{x}{r_\eta^{\mathrm{ns}}}|
    \le
    2^{\kappa_\eta^{\mathrm{ns}}n+o(n)}.
    \]
\end{enumerate}
\end{theorem}

\begin{proof}
Deferred to Section~\ref{sec:technical-details}.
\end{proof}

\Cref{th:ns-favorable} identifies the exponent in the correlated-pair regime:
\[
\boxed{
\kappa_\eta^{\mathrm{ns}}=h(q_\eta).
}
\]
This exponent depends only on \(\eta\) and the degree \(k\), and not on the normalized optimum scale. That is exactly what makes this construction natural in the given-threshold setting.

% ==================================================
\subsection{Local-Lipschitz analysis}\label{subsec:local-lipschitz-analysis}

We now turn to the second regime. This is the setting relevant to Case~2 of Section~\ref{subsec:threshold-exponent-cases}, where the threshold exponent is derived from instance parameters. Here dependence on the normalized optimum scale is already present on the threshold side, and we exploit this by working with weighted MAX-\(k\)-CSP instances.

A weighted MAX-\(k\)-CSP instance on variables
\[
x_1,\dots,x_n\in\{-1,1\}
\]
consists of constraints indexed by \(j\in[m]\). 
% For each constraint \(j\), let
% \[
% \vars(j)\subseteq [n]
% \qquad\text{with}\qquad
% k_j:=|\vars(j)|\le k.
% \]
The \(j\)-th constraint is specified by
\begin{itemize}
    \item a set of variables \(\vars(j)\subseteq[n]\) of size \(k_j\le k\),
%    \item a literal pattern on these variables,
    \item a positive rational weight \(w_j\in\mathbb Q_{>0}\),
    \item a nontrivial local predicate
    \[
    P_j:\{-1,1\}^{k_j}\to\{0,1\}.
    \]
\end{itemize}
We write
\[
s_j=|P_j^{-1}(1)|\in[1,2^{k_j}-1].
\]

Define the centered contribution of clause \(j\) by
\[
C_j(x)=
\begin{cases}
-w_j & \text{if the \(j\)-th constraint is satisfied by \(x\),}\\[2mm]
\dfrac{s_j}{2^{k_j}-s_j}\,w_j & \text{if the \(j\)-th constraint is violated by \(x\).}
\end{cases}
\]
The centered objective is
\[
H(x)=\sum_{j=1}^m C_j(x).
\]

Because each local input is uniform under the uniform distribution on \(\{-1,1\}^n\), each centered contribution \(C_j\) has mean zero. Therefore
\[
\mathbb E[H]=0.
\]

We write
\[
W=\sum_{j=1}^m w_j,
\qquad
H_{\mathrm{min}}=\min_{x\in\{-1,1\}^n} H(x)\le 0,
\]
and fix an optimum assignment \(x^\ast\in\{-1,1\}^n\).

Since each summand satisfies \(C_j(x)\ge -w_j\), one has
\[
H(x)\ge -W
\qquad\text{for all }x\in\{-1,1\}^n.
\]
Hence
\[
0\le |H_{\mathrm{min}}|\le W.
\]

For \(\eta\in(0,1)\), define the threshold set
\[
T_\eta=\{x\in\{-1,1\}^n:H(x)\le (1-\eta)H_{\mathrm{min}}\}.
\]

For each variable \(i\in[n]\), define its weighted degree by
\[
d_i=\sum_{j:\,i\in \vars(j)} w_j.
\]
We also define the total incident weight
\[
\Sigma=\sum_{i=1}^n d_i=\sum_{j=1}^m k_j w_j,
\]
the average weighted degree
\[
d_{\mathrm{avg}}=\frac{\Sigma}{n},
\]
and the irregularity parameter
\[
D=\frac{n\sum_{i=1}^n d_i^2}{\Sigma^2}.
\]
By the Cauchy--Schwarz inequality,
\[
\sum_{i=1}^n d_i^2\ge \frac{1}{n}\left(\sum_{i=1}^n d_i\right)^2
= \frac{\Sigma^2}{n},
\]
and therefore
\[
D\ge 1.
\]

Since \(k_j\le k\) for every \(j\), we also have
\[
\Sigma=\sum_{j=1}^m k_j w_j\le k\sum_{j=1}^m w_j = kW.
\]
This inequality will later enable us to derive several coarse bounds. 

%Thus, for the coarse bounds derived later, the worst case occurs when all constraints have arity exactly \(k\).

For each clause \(j\), define
\[
\Lambda_j=\frac{2^{k_j}}{2^{k_j}-s_j},
\qquad
\Lambda_{\max}=\max_{j\in[m]}\Lambda_j.
\]

Let \(d_H\) denote the Hamming distance on \(\{-1,1\}^n\). For \(L\subseteq[n]\), \(x\in\{-1,1\}^n\), and \(r\ge 0\), define the restricted Hamming ball
\[
\LightBall{L}{x}{r}
=
\{y\in\{-1,1\}^n:\ \{i\in[n]:x_i\neq y_i\}\subseteq L,\ d_H(x,y)\le r\}.
\]

Consider the set
\[
L_{\av}=\{i\in[n]:d_i\le 2d_{\mathrm{avg}}\}.
\]
For conciseness, we introduce the notation \(\LightBall{\av}{x}{r}\) for
\(\LightBall{L_{\av}}{x}{r}\).
The algorithmic point is that, on the coordinates in \(L_{\av}\), a bounded number of flips changes the objective by only a controlled amount.
Formalizing this argument leads to the following statement.

\begin{proposition}[Local-Lipschitz bound on light coordinates]\label{prop:lip-local}
For every integer \(r\ge 0\) and every
\[
x\in \LightBall{\av}{x^\ast}{r},
\]
we have
\[
H(x)\le H_{\mathrm{min}}+2\Lambda_{\max} d_{\mathrm{avg}}\,r.
\]
\end{proposition}
% \begin{proposition}[Local-Lipschitz bound on light coordinates]\label{prop:lip-local}
% For every integer \(r\ge 0\) and every
% \[
% x\in \LightBall{\av}{x^\ast}{r},
% \]
% we have
% \[
% H(x)\le H_{\mathrm{min}}+2\Lambda_{\max} d_{\mathrm{avg}}\,r.
% \]
% \end{proposition}

\begin{proof}
Starting from \(x^\ast\), flip the differing coordinates one by one. If one flips a coordinate \(i\in L_{\av}\), only constraints containing \(i\) can change. For any such constraint \(j\), the two possible values of \(C_j\) differ by exactly
\[
w_j+\frac{s_j}{2^{k_j}-s_j}w_j
=
\Lambda_j w_j
\le
\Lambda_{\max} w_j.
\]
Hence the total change caused by flipping \(i\) is at most
\[
\sum_{j:\,i\in\vars(j)} \Lambda_j w_j
\le
\Lambda_{\max}\sum_{j:\,i\in\vars(j)} w_j
=
\Lambda_{\max} d_i
\le
2\Lambda_{\max} d_{\mathrm{avg}}.
\]
Summing over at most \(r\) flips proves the claim.
\end{proof}
% \begin{proof}
% Starting from \(x^\ast\), flip the differing coordinates one by one. If one flips a coordinate \(i\in L_{\av}\), only constraints containing \(i\) can change. For any such constraint \(j\), the two possible values of \(C_j\) differ by exactly
% \[
% w_j+\frac{s_j}{2^k-s_j}w_j
% =
% \Lambda_j w_j
% \le
% \Lambda_{\max} w_j.
% \]
% Hence the total change caused by flipping \(i\) is at most
% \[
% \sum_{j:\,i\in\vars(j)} \Lambda_j w_j
% \le
% \Lambda_{\max}\sum_{j:\,i\in\vars(j)} w_j
% =
% \Lambda_{\max} d_i
% \le
% 2\Lambda_{\max} d_{\mathrm{avg}}.
% \]
% Summing over at most \(r\) flips proves the claim.
% \end{proof}

We now choose the radius so that this additive error stays within the threshold slack \(\eta |H_{\mathrm{min}}|\). Define
\[
\theta_\eta
=
\frac{\eta |H_{\mathrm{min}}|}{\Lambda_{\max}\Sigma},
\qquad
r_\eta^{\mathrm{lip}}
=
\left\lfloor
\frac{\eta |H_{\mathrm{min}}|}{2\Lambda_{\max} d_{\mathrm{avg}}}
\right\rfloor
=
\left\lfloor \frac{\theta_\eta n}{2}\right\rfloor,
\]
and
\[
S_\eta^{\mathrm{lip}}
=
\LightBall{\av}{x^\ast}{r_\eta^{\mathrm{lip}}}.
\]

Since \(\Sigma\le kW\), we also have the coarse lower bound
\[
\theta_\eta
\ge
\frac{\eta}{\Lambda_{\max}k}\cdot \frac{|H_{\mathrm{min}}|}{W}.
\]
Thus the smallest value of \(\theta_\eta\) compatible with the parameters \(k\), \(W\), \(\Lambda_{\max}\), and \(|H_{\mathrm{min}}|\) is attained when all constraints have arity exactly \(k\).

% We now choose the radius so that this additive error stays within the threshold slack \(\eta |H_{\mathrm{min}}|\). Define
% \[
% \theta_\eta
% =
% \frac{\eta}{\Lambda_{\max}k}\cdot \frac{|H_{\mathrm{min}}|}{W}
% =
% \frac{\eta(2^k-s_{\min})}{2^k k}\cdot \frac{|H_{\mathrm{min}}|}{W},
% \qquad
% r_\eta^{\mathrm{lip}}
% =
% \left\lfloor
% \frac{\eta |H_{\mathrm{min}}|}{2\Lambda_{\max} d_{\mathrm{avg}}}
% \right\rfloor
% =
% \left\lfloor \frac{\theta_\eta n}{2}\right\rfloor,
% \]
% and
% \[
% S_\eta^{\mathrm{lip}}
% =
% \LightBall{\av}{x^\ast}{r_\eta^{\mathrm{lip}}}.
% \]

\begin{theorem}[Successful-set bound in the local-Lipschitz regime]\label{th:lip-favorable}
Fix \(\eta\in(0,1)\). Then the set \(S_\eta^{\mathrm{lip}}\) satisfies the following.
\begin{enumerate}
    \item \(S_\eta^{\mathrm{lip}}\subseteq T_\eta\).
    \item For every \(x\in S_\eta^{\mathrm{lip}}\), the restricted Hamming ball
    \[
    \LightBall{\av}{x}{r_\eta^{\mathrm{lip}}}
    \]
    contains an optimum assignment.
    \item One has
    \[
    |S_\eta^{\mathrm{lip}}|
    \ge
    2^{\kappa_\eta^{\mathrm{lip}}n-o(n)},
    \qquad
    \kappa_\eta^{\mathrm{lip}}=\frac12\,h(\theta_\eta).
    \]
    \item For every \(x\in\{-1,1\}^n\),
    \[
    \left|
    \LightBall{\av}{x}{r_\eta^{\mathrm{lip}}}
    \right|
    =
    |S_\eta^{\mathrm{lip}}|.
    \]
    In particular, the corresponding region can be searched in time
    \[
    O^*(|S_\eta^{\mathrm{lip}}|).
    \]
\end{enumerate}
\end{theorem}

% \begin{theorem}[Successful-set bound in the local-Lipschitz regime]\label{th:lip-favorable}
% Fix \(\eta\in(0,1)\). Then the set \(S_\eta^{\mathrm{lip}}\) satisfies the following.
% \begin{enumerate}
%     \item \(S_\eta^{\mathrm{lip}}\subseteq T_\eta\).
%     \item For every \(x\in S_\eta^{\mathrm{lip}}\), the restricted Hamming ball
%     \[
%     \LightBall{\av}{x}{r_\eta^{\mathrm{lip}}}
%     \]
%     contains an optimum assignment.
%     \item One has
%     \[
%     |S_\eta^{\mathrm{lip}}|
%     \ge
%     2^{\kappa_\eta^{\mathrm{lip}}n-o(n)},
%     \qquad
%     \kappa_\eta^{\mathrm{lip}}=\frac12\,h(\theta_\eta).
%     \]
%     \item For every \(x\in\{-1,1\}^n\),
%     \[
%     \left|
%     \LightBall{\av}{x}{r_\eta^{\mathrm{lip}}}
%     \right|
%     =
%     |S_\eta^{\mathrm{lip}}|.
%     \]
%     In particular, the corresponding region can be searched in time
%     \[
%     O^*(|S_\eta^{\mathrm{lip}}|).
%     \]
% \end{enumerate}
% \end{theorem}

\begin{proof}
Deferred to Section~\ref{sec:technical-details}.
\end{proof}

% \begin{remark}\label{rem:lip-worst-case}
% Since \(\Sigma\le kW\), the parameter \(\theta_\eta\) satisfies
% \[
% \theta_\eta
% \ge
% \frac{\eta}{\Lambda_{\max}k}\cdot \frac{|H_{\mathrm{min}}|}{W}.
% \]
% Hence the corresponding successful-set exponent \(\kappa_\eta^{\mathrm{lip}}=\frac12 h(\theta_\eta)\) is never smaller than the value obtained in the exact-arity-\(k\) case.
% \end{remark}

\Cref{th:lip-favorable} characterizes the exponent in the local-Lipschitz regime:
\[
\boxed{
\kappa_\eta^{\mathrm{lip}}=\frac12\,h(\theta_\eta).
}
\]
Unlike the correlated-pair exponent \(h(q_\eta)\), this quantity depends explicitly on the ratio \(|H_{\mathrm{min}}|/\Sigma\). This is natural in the present regime, because the threshold exponent derived from concentration bounds also depends on \(|H_{\mathrm{min}}|/\Sigma\).
%This is natural in the present regime, because the threshold exponent derived from concentration bounds depends on the same scale
% Unlike the correlated-pair exponent \(h(q_\eta)\), this quantity depends explicitly on the ratio \(|H_{\mathrm{min}}|/\Sigma\). In the present regime this is natural, because the threshold exponent itself will already depend on the same scale when derived from concentration bounds.
% Unlike the correlated-pair exponent \(h(q_\eta)\), this quantity depends explicitly on the normalized optimum scale \(|H_{\mathrm{min}}|/W\). In the present regime this is natural, because the threshold exponent itself will already depend on the same scale when derived from concentration bounds.

% ==================================================
\subsection*{Summary of Section~\ref{sec:favorable-sets}}

This section has provided the exponents needed in the black-box theorem of Section~\ref{sec:framework}:
\[
\kappa_\eta=
\begin{cases}
\kappa_\eta^{\mathrm{ns}}=h(q_\eta)
& \text{in the correlated-pair regime,}\\[1mm]
\kappa_\eta^{\mathrm{lip}}=\dfrac12 h(\theta_\eta)
& \text{in the local-Lipschitz regime, where }
\theta_\eta=\dfrac{\eta |H_{\mathrm{min}}|}{\Lambda_{\max}\Sigma}.
\end{cases}
\]
% \[
% \kappa_\eta=
% \begin{cases}
% \kappa_\eta^{\mathrm{ns}}=h(q_\eta)
% & \text{in the correlated-pair regime,}\\[1mm]
% \kappa_\eta^{\mathrm{lip}}=\dfrac12\,h(\theta_\eta)
% & \text{in the local-Lipschitz regime.}
% \end{cases}
% \]
The next section combines these successful-set constructions with threshold exponents \(\gamma_\eta\) to obtain concrete optimization algorithms.

% ==================================================
% Section 4: Concrete corollaries
% ==================================================
\section{Proof of the Main Results}\label{sec:concrete-corollaries}
In this section, we prove Theorems \ref{th6c}, \ref{th7c}, \ref{th6q2} and \ref{th7q2}. This is done by combining the black-box theorem of Section~\ref{sec:framework} with the successful-set constructions of Section~\ref{sec:favorable-sets}. This yields two results corresponding to the two cases described in Section~\ref{subsec:threshold-exponent-cases}.

The first case, discussed in \Cref{subsec:concrete-eklin2}, is the one in which a threshold-set exponent \((\gamma,\eta)\) is assumed as part of the input hypothesis. This leads naturally to MAX-E\(k\)-LIN2, where the successful-set exponent comes from the correlated-pair analysis of Section~\ref{subsec:noise-stability-analysis} and does not depend on the normalized optimum scale. This yields \Cref{th6c}. 

The second case, discussed in \Cref{subsec:concrete-maxkcsp}, is the one in which the threshold-set exponent is derived from instance parameters. In this regime, dependence on \(|H_{\mathrm{min}}|/\Sigma\) is already unavoidable in the upper bound on the threshold set. We therefore use the local-Lipschitz analysis of Section~\ref{subsec:local-lipschitz-analysis}, which yields a result for weighted MAX-\(k\)-CSP. This yields \Cref{th7c}.

Finally, \Cref{subsec:brief-quantum} explains how to derive \Cref{th6q2,th7q2}.

% ==================================================
\subsection[Case 1: given $(\gamma,\eta)$ --- MAX-E$k$-LIN2
]{Case 1: given \texorpdfstring{$\boldsymbol{(\gamma,\eta)}$}{(gamma, eta)} --- MAX-E\texorpdfstring{$k$}{k}-LIN2}
\label{subsec:concrete-eklin2}

We first consider the setting in which a threshold-set exponent is assumed as part of the input hypothesis:
\[
|T_\eta|\le 2^{(1-\gamma)n},
\qquad
T_\eta=\{x:H(x)\le (1-\eta)H_{\mathrm{min}}\}.
\]
In this case, the only remaining task is to construct a large successful subset of \(T_\eta\) and compute the corresponding exponent.

For aggregated weighted MAX-E\(k\)-LIN2 instances, maximizing the satisfied weight is equivalent to minimizing a canonical centered homogeneous degree-\(k\) polynomial \(H\) on \(\{-1,1\}^n\). Thus, the homogeneous degree-\(k\) framework of Section~\ref{subsec:noise-stability-analysis} applies directly to the centered objective canonically associated with the MAX-E\(k\)-LIN2 instance.

By \Cref{th:ns-favorable}, for every fixed \(\eta\in(0,1)\), there exists an explicit set
\[
S_\eta^{\mathrm{ns}}\subseteq T_\eta
\]
such that
\[
|S_\eta^{\mathrm{ns}}|
\ge
2^{\kappa_\eta^{\mathrm{ns}}n-o(n)},
\qquad
\kappa_\eta^{\mathrm{ns}}=h(q_\eta),
\qquad
q_\eta=\frac{1-(1-\eta)^{1/k}}{2},
\]
and every point of \(S_\eta^{\mathrm{ns}}\) comes with a Hamming ball containing an optimum assignment.

Applying the black-box theorem yields the following result. 

\begin{theorem}[Main theorem for MAX-E\(k\)-LIN2]\label{th:concrete-eklin2}
Assume that the optimum value \(H_{\mathrm{min}}\) is given. Fix \(\eta\in(0,1)\) and \(\gamma>0\). Suppose that
\[
|T_\eta|\le 2^{(1-\gamma)n},
\qquad
T_\eta=\{x:H(x)\le (1-\eta)H_{\mathrm{min}}\}.
\]
Then an optimum assignment can be found with high probability in time
$
2^{(1-c)n+o(n)}
$ with
\[
c=
\min\{\gamma,h(q_\eta)\}\,,
\]
where
\[
q_\eta=\frac{1-(1-\eta)^{1/k}}{2}.
\]
\end{theorem}
\begin{proof}%[Proof of \Cref{th:concrete-eklin2}]
Combine \Cref{th:black-box} with \Cref{th:ns-favorable}, and then convert the expected time complexity into a bounded-error complexity using Markov inequality. The threshold-set exponent is \(\gamma\), and the successful-set exponent is \(\kappa_\eta^{\mathrm{ns}}=h(q_\eta)\).
\end{proof}
\begin{remark}\label{rem:1}
The assumption that the optimum value \(H_{\mathrm{min}}\) is given is made to simplify the presentation. It is used only to implement the membership test in \(T_\eta\). For unweighted MAX-E\(k\)-LIN2, this assumption can be easily removed, since there are only polynomially many possible values for \(H_{\mathrm{min}}\). As explained in Appendices \ref{sec:standard-variants} and \ref{sec:top-k-variant}, this assumption can also be removed in the weighted case, with at most a $2^{o(n)}$ multiplicative overhead.
\end{remark}

We obtain \Cref{th6c} stated in the introduction by using the inequality
\[
q_\eta\ge \frac{\eta}{2k},
\]
which was established in Section~\ref{subsec:noise-stability-analysis}.
% We obtain \Cref{th6c} stated in the introduction by using the inequality 
% \[
% q_\eta\ge \frac{\eta}{2k}\,.
% \]

\begin{remark}
\Cref{th:concrete-eklin2} is exactly Case~1 of Section~\ref{subsec:threshold-exponent-cases}: the exponent \(\gamma\) is supplied by hypothesis, and the final classical exponent is obtained by comparing \(\gamma\) with the successful-set exponent \(h(q_\eta)\).
The salient feature of this regime is that the successful-set exponent is determined entirely by homogeneous degree-\(k\) noise stability on the Boolean cube. In particular, once \((\gamma,\eta)\) is given, the resulting exponent
\[
\min\{\gamma,h(q_\eta)\}
\]
does not depend on the normalized optimum scale.
\end{remark}
% ==================================================
\subsection[Case 2: instance-determined $(\gamma_\eta,\eta)$ --- weighted MAX-$k$-CSP
]{Case 2: instance-determined \texorpdfstring{$\boldsymbol{(\gamma_\eta,\eta)}$}{(gamma_eta, eta)} --- weighted MAX-\texorpdfstring{$\boldsymbol{k}$}{k}-CSP}
\label{subsec:concrete-maxkcsp}

We now turn to the setting in which the threshold-set exponent is derived from the instance itself. This is the natural regime for the local-Lipschitz analysis of weighted MAX-\(k\)-CSP.

%The point of using the local-Lipschitz framework here is the following. In Case 1, dependence of the successful-set exponent on the normalized optimum scale \(|H_{\mathrm{min}}|/W\) would be undesirable, since the sublevel-set exponent would already be supplied externally. In the present case, however, such dependence is already present in the concentration bound that upper bounds the conditioning set. It is therefore natural to allow the successful-set exponent to depend on the same scale, and in return to obtain a result for weighted MAX-\(k\)-CSPs rather than only for homogeneous degree-\(k\) objectives.

By \Cref{th:lip-favorable}, for every fixed \(\eta\in(0,1)\), there exists an explicit set
\[
S_\eta^{\mathrm{lip}}\subseteq T_\eta
\]
such that
\[
|S_\eta^{\mathrm{lip}}|
\ge
2^{\kappa_\eta^{\mathrm{lip}}n-o(n)},
\qquad
\kappa_\eta^{\mathrm{lip}}=\frac12\,h(\theta_\eta),
\]
where
\[
\theta_\eta
=
\frac{\eta |H_{\mathrm{min}}|}{\Lambda_{\max}\Sigma}.
\]
% where
% \[
% \theta_\eta
% =
% \frac{\eta}{\Lambda_{\max}k}\cdot \frac{|H_{\mathrm{min}}|}{W}.
% % =
% % \frac{\eta(2^k-s_{\min})}{2^k k}\cdot \frac{|H_{\mathrm{min}}|}{W}.
% \]

On the other hand, a McDiarmid bound gives an upper bound on $|T_\eta|$:

\begin{proposition}[McDiarmid threshold-set bound]\label{prop:tech-mcdiarmid}
For every \(\eta\in(0,1)\),
\[
|T_\eta|\le 2^{(1-\gamma_\eta)n},
\]
where
\[
\gamma_\eta
=
\frac{2(1-\eta)^2}{\ln 2}\cdot
\frac{1}{\Lambda_{\max}^2 D}
\left(\frac{|H_{\mathrm{min}}|}{\Sigma}\right)^2.
\]
\end{proposition}
% \begin{proposition}[McDiarmid threshold-set bound]\label{prop:tech-mcdiarmid}
% For every \(\eta\in(0,1)\),
% \[
% |T_\eta|\le 2^{(1-\gamma_\eta)n},
% \]
% where
% \[
% \gamma_\eta
% =
% \frac{2(1-\eta)^2}{\ln 2}\cdot
% \frac{1}{\Lambda_{\max}^2 k^2D}
% \left(\frac{|H_{\mathrm{min}}|}{W}\right)^2.
% % =
% % \frac{2(1-\eta)^2}{\ln 2}\cdot
% % \frac{(2^k-s_{\min})^2}{2^{2k}k^2D}
% % \left(\frac{|H_{\mathrm{min}}|}{W}\right)^2.
% \]
% \end{proposition}

\begin{proof}
Under the uniform distribution on \(\{-1,1\}^n\), one has \(\mathbb E[H]=0\). For each coordinate \(i\), changing the sign of \(x_i\) can affect only constraints containing \(i\). For each such affected constraint~\(j\), the contribution changes by at most
\[
\Lambda_j w_j\le \Lambda_{\max}w_j.
\]
Hence McDiarmid's inequality applies with bounded-difference parameters
\[
\mDelta_i=\sum_{j:\,i\in\vars(j)} \Lambda_j w_j
\le \Lambda_{\max} d_i.
\]
Therefore
\[
\sum_{i=1}^n \mDelta_i^2
\le
\Lambda_{\max}^2\sum_{i=1}^n d_i^2
=
\Lambda_{\max}^2\frac{D\Sigma^2}{n}.
\]

The event \(x\in T_\eta\) is exactly
\[
H(x)\le (1-\eta)H_{\mathrm{min}}=-(1-\eta)|H_{\mathrm{min}}|.
\]
McDiarmid's inequality yields
\[
\Pr[x\in T_\eta]
\le
\exp\!\left(
-\frac{2(1-\eta)^2|H_{\mathrm{min}}|^2}{\sum_{i=1}^n\mDelta_i^2}
\right)
\le
\exp\!\left(
-\frac{2(1-\eta)^2|H_{\mathrm{min}}|^2 n}{\Lambda_{\max}^2 D\Sigma^2}
\right).
\]
Converting to base \(2\) and multiplying by \(2^n\) assignments gives the stated bound.
\end{proof}
% \begin{proof}
% Under the uniform distribution on \(\{-1,1\}^n\), one has \(\mathbb E[H]=0\). For each coordinate \(i\), changing the sign of \(x_i\) can affect only constraints containing \(i\). For each such affected constraint~\(j\), the contribution changes by at most
% \[
% \Lambda_j w_j\le \Lambda_{\max}w_j.
% \]
% Hence McDiarmid's inequality applies with bounded-difference parameters
% \[
% \Lambda_i=\sum_{j:\,i\in\vars(j)} \Lambda_j w_j
% \le \Lambda_{\max} d_i.
% \]
% Therefore
% \[
% \sum_{i=1}^n \Lambda_i^2
% \le
% \Lambda_{\max}^2\sum_{i=1}^n d_i^2
% =
% \Lambda_{\max}^2 k^2W^2D/n.
% \]

% The event \(x\in T_\eta\) is exactly
% \[
% H(x)\le (1-\eta)H_{\mathrm{min}}=-(1-\eta)|H_{\mathrm{min}}|.
% \]
% McDiarmid's inequality yields
% \[
% \Pr[x\in T_\eta]
% \le
% \exp\!\left(
% -\frac{2(1-\eta)^2|H_{\mathrm{min}}|^2}{\sum_{i=1}^n\Lambda_i^2}
% \right)
% \le
% \exp\!\left(
% -\frac{2(1-\eta)^2|H_{\mathrm{min}}|^2n}{\Lambda_{\max}^2k^2W^2D}
% \right).
% \]
% Converting to base \(2\) and multiplying by \(2^n\) assignments gives the stated bound.
% \end{proof}

Combining these two ingredients with \Cref{th:black-box} yields the following.

\begin{theorem}[Main result for weighted MAX-\(k\)-CSP]\label{th:concrete-maxkcsp}
Assume that the optimum value \(H_{\mathrm{min}}\) is given. Fix \(\eta\in(0,1)\). For every weighted MAX-\(k\)-CSP instance, an optimum assignment can be found with high probability in time
\[
2^{(1-c)n+o(n)}
\]
with
\[
c
=
\min\!\left\{
\gamma_\eta,
\;
\frac12\,h(\theta_\eta)
\right\},
\]
where
\begin{align*}
\gamma_\eta
&=
\frac{2(1-\eta)^2}{\ln 2}\cdot
\frac{1}{\Lambda_{\max}^2 D}
\left(\frac{|H_{\mathrm{min}}|}{\Sigma}\right)^2,\\
\theta_\eta
&=
\frac{\eta |H_{\mathrm{min}}|}{\Lambda_{\max}\Sigma}.
\end{align*}
\end{theorem}
% \begin{theorem}[Main result for weighted MAX-\(k\)-CSP]\label{th:concrete-maxkcsp}
% Assume that the optimum value \(H_{\mathrm{min}}\) is given. Fix \(\eta\in(0,1)\). For every weighted MAX-\(k\)-CSP instance, an optimum assignment can be found with high probability in time
% $
% %2^{(1-c^{\mathrm{cl}}_{7}(\eta))n+o(n)},
% 2^{(1-c)n+o(n)}
% $
% with
% \[
% %c^{\mathrm{cl}}_{7}(\eta)
% c
% =
% \min\!\left\{
% \gamma_\eta
% ,
% \;
% \frac12\,h(\theta_\eta)
% \right\}\,,
% \]
% where
% \begin{align*}
% \gamma_\eta
% &=
% \frac{2(1-\eta)^2}{\ln 2}\cdot
% \frac{1}{\Lambda_{\max}^2 k^2D}
% \left(\frac{|H_{\mathrm{min}}|}{W}\right)^2,\\
% \theta_\eta
% &=
% \frac{\eta}{\Lambda_{\max}k}\cdot \frac{|H_{\mathrm{min}}|}{W}\,.
% \end{align*}
% \end{theorem}

\begin{proof}[Proof of \Cref{th:concrete-maxkcsp}]
Combine \Cref{th:black-box}, \Cref{th:lip-favorable}, and \Cref{prop:tech-mcdiarmid}. The threshold-set exponent is \(\gamma_\eta\), and the successful-set exponent is \(\kappa_\eta^{\mathrm{lip}}=\frac12 h(\theta_\eta)\). Finally, convert the resulting algorithm (which has guaranteed expected time complexity) into an algorithm with bounded-error time complexity using Markov inequality.
\end{proof}
% \begin{proof}[Proof of \Cref{th:concrete-maxkcsp}]
% Combine \Cref{th:black-box}, \Cref{th:lip-favorable}, and \Cref{prop:tech-mcdiarmid}. The threshold-set exponent is \(\gamma_\eta\), and the successful-set exponent is \(\kappa_\eta^{\mathrm{lip}}=\frac12 h(\theta_\eta)\). Finally, convert the resulting algorithm (which has guaranteed expected time complexity) into an algorithm with bounded-error time complexity using Markov inequality.
% \end{proof}
\begin{remark}\label{rem:2}
The assumption that the optimum value \(H_{\mathrm{min}}\) is given is included only to keep the main statement as explicit as possible. 
For unweighted MAX-\(k\)-CSP, this assumption can again be easily removed, since there are only polynomially many possible values for \(H_{\mathrm{min}}\). As explained in Appendices \ref{sec:standard-variants} and \ref{sec:optimistic-first-local-ball}, this assumption can also be removed in the weighted case with at most a $2^{o(n)}$ multiplicative overhead.
\end{remark}

The following particularly simple explicit consequence is obtained by fixing \(\eta=1/2\). 
%%%%%%%%%%%%%%%%%%%%%%%%%%%%%%%%%%%%%%%%%%%%%%%%%%
%%%%%%%%%%%%%%%%%%%%%%%%%%%%%%%%%%%%%%%%%%%%%%%%%%
%%%%%%%%%%%%%%%%%%%%%%%%%%%%%%%%%%%%%%%%%%%%%%%%%%
\begin{corollary}[Explicit consequence for weighted MAX-\(k\)-CSP]\label{cor:concrete-maxkcsp-explicit}
For every weighted MAX-\(k\)-CSP instance, an optimum assignment can be found with high probability in time
\[
2^{(1-c)n+o(n)}
\]
with
\[
c
\ge
\frac{1}{2\ln 2}\cdot
\frac{1}{\Lambda_{\max}^2 k^2D}
\left(\frac{|H_{\mathrm{min}}|}{W}\right)^2
\ge
\frac{1}{2\ln 2}\cdot
\frac{1}{2^{2k}k^2D}
\left(\frac{|H_{\mathrm{min}}|}{W}\right)^2.
\]
\end{corollary}

\begin{proof}
Evaluate \Cref{th:concrete-maxkcsp} at \(\eta=1/2\). Then
\[
c
=
\min\!\left\{
\gamma_{1/2},
\frac12 h(\theta_{1/2})
\right\},
\]
where
\[
\gamma_{1/2}
=
\frac{1}{2\ln 2}\cdot
\frac{1}{\Lambda_{\max}^2 D}
\left(\frac{|H_{\mathrm{min}}|}{\Sigma}\right)^2,
\qquad
\theta_{1/2}
=
\frac{|H_{\mathrm{min}}|}{2\Lambda_{\max}\Sigma}.
\]

Since \(\Sigma\le kW\), we have
\[
\gamma_{1/2}
\ge
\widetilde{\gamma}_{1/2}
:=
\frac{1}{2\ln 2}\cdot
\frac{1}{\Lambda_{\max}^2 k^2D}
\left(\frac{|H_{\mathrm{min}}|}{W}\right)^2,
\]
and
\[
\theta_{1/2}
\ge
\widetilde{\theta}_{1/2}
:=
\frac{|H_{\mathrm{min}}|}{2\Lambda_{\max}kW}.
\]
Since \(0\le \widetilde{\theta}_{1/2}\le 1/2\) and the binary entropy function is increasing on \([0,1/2]\), it follows that
\[
\frac12 h(\theta_{1/2})
\ge
\frac12 h(\widetilde{\theta}_{1/2}).
\]
Therefore
\[
c
\ge
\min\!\left\{
\widetilde{\gamma}_{1/2},
\frac12 h(\widetilde{\theta}_{1/2})
\right\}.
\]

It remains to show that
\[
\widetilde{\gamma}_{1/2}\le \frac12 h(\widetilde{\theta}_{1/2}).
\]
Set
\[
\xi=\frac{|H_{\mathrm{min}}|}{W},
\qquad
A=\frac{1}{2\Lambda_{\max}k},
\qquad
B=\frac{1}{2\ln 2}\cdot\frac{1}{\Lambda_{\max}^2k^2D}.
\]
Then
\[
\widetilde{\theta}_{1/2}=A\xi,
\qquad
\widetilde{\gamma}_{1/2}=B\xi^2.
\]
Since \(0<\xi\le 1\), \(D\ge 1\), and \(A\xi\in(0,1/2]\), one has
\[
\frac12 h(A\xi)\ge A\xi
\]
because \(h(x)\ge 2x\) on \([0,1/2]\).

Moreover,
\[
\frac{B}{A}
=
\frac{1}{\ln 2\,\Lambda_{\max}kD}.
\]
If \(k\ge 2\), then \(\Lambda_{\max}\ge 1\), and hence
\[
\frac{B}{A}
\le
\frac{1}{2\ln 2\,D}
\le 1.
\]
If \(k=1\), then every nontrivial clause has arity \(k_j=1\) and therefore
\[
\Lambda_j=\frac{2}{2-1}=2
\qquad\text{for every }j,
\]
so \(\Lambda_{\max}=2\). Hence again
\[
\frac{B}{A}
=
\frac{1}{2\ln 2\,D}
\le 1.
\]
Thus in all cases
\[
\frac{B}{A}\le 1.
\]
Hence
\[
B\xi^2\le A\xi\le \frac12 h(A\xi),
\]
which proves that
\[
\widetilde{\gamma}_{1/2}\le \frac12 h(\widetilde{\theta}_{1/2}).
\]
Therefore
\[
c\ge \widetilde{\gamma}_{1/2}
=
\frac{1}{2\ln 2}\cdot
\frac{1}{\Lambda_{\max}^2 k^2D}
\left(\frac{|H_{\mathrm{min}}|}{W}\right)^2.
\]

Finally, since \(\Lambda_{\max}\le 2^k\), we obtain
\[
c
\ge
\frac{1}{2\ln 2}\cdot
\frac{1}{2^{2k}k^2D}
\left(\frac{|H_{\mathrm{min}}|}{W}\right)^2.
\]
\end{proof}

We obtain \Cref{th7c} stated (for unweighted MAX-\(k\)-CSP) in the introduction by observing that %$\Delta=|H_{\mathrm{min}}|/W$ 
\[
\frac{1}{2\ln 2}\ge 0.72134\,.
\]

\subsection{Quantum algorithms}\label{subsec:brief-quantum}
\Cref{th6q2} and \Cref{th7q2}, stated in the introduction, follow by applying quantum search to obtain a quadratic speedup of the algorithms presented in Sections~\ref{subsec:concrete-eklin2} and~\ref{subsec:concrete-maxkcsp}. We only sketch the argument, since it is a routine combination of amplitude amplification~\cite{BrassardHMT02} and quantum minimum finding~\cite{Durr1999}.

We first derive a quantum version of \Cref{cor:abstract-conditioning-search-rejection}.
Let \(\Omega=\{-1,1\}^n\) be the ambient space. As in Section~\ref{sec:framework}, we denote by
 \(T\subseteq \Omega\) the conditioning set, and by \(S\subseteq T\) the successful set. As in \Cref{cor:abstract-conditioning-search-rejection}, we assume that membership in~\(T\) can be decided in (classical) polynomial time. We now assume that there is a quantum algorithm that, for any \(x\in T\), searches the neighborhood \(\mathcal N(x)\) in time at most \(\tau_{\mathrm{search}}^{\mathrm{q}}(n)\). 
Then the arguments of \Cref{sub:abstract}, combined with amplitude amplification \cite{BrassardHMT02}, yield a quantum algorithm that finds a solution in time (omitting factors polynomial in $n$)
\[
%O\!\left(
\sqrt{\frac{\abs{\Omega}}{\abs{S}}}
+
\sqrt{\frac{\abs{T}}{\abs{S}}}\cdot\tau_{\mathrm{search}}^{\mathrm{q}}(n)\,.
%\right).
\]
Compared to \Cref{cor:abstract-conditioning-search-rejection}, we obtain a quadratic speedup when \(\tau_{\mathrm{search}}^{\mathrm{q}}(n)\) is quadratically faster than its classical counterpart \(\tau_{\mathrm{search}}(n)\).

In the applications of this framework given in Sections~\ref{subsec:concrete-eklin2} and \ref{subsec:concrete-maxkcsp}, \(\tau_{\mathrm{search}}(n)\) is the cost of an exhaustive search for a minimum value in the neighborhood. In the quantum setting, we obtain a quadratic speedup by using the quantum algorithm for minimum finding \cite{Durr1999}. The arguments of Sections~\ref{subsec:concrete-eklin2} and \ref{subsec:concrete-maxkcsp} therefore yield \Cref{th6q2} and \Cref{th7q2}.

% ==================================================
% \subsection*{Summary of Section~\ref{sec:concrete-corollaries}}

% The two corollaries of this section instantiate the same black-box theorem in two different ways:
% \[
% \begin{array}{c|c|c}
% & \text{Case 1} & \text{Case 2}\\
% \hline
% \text{sublevel-set bound} & (\gamma,\eta) \text{ given}& \text{derived from instance parameters}\\[1mm]
% \text{successful-set bound} & \kappa_\eta^{\mathrm{ns}}=h(q_\eta) & \kappa_\eta^{\mathrm{lip}}=\frac12 h(\theta_\eta)\\[1mm]
% \text{setting} & \text{MAX-E}k\text{-LIN2} & \text{weighted MAX-}k\text{-CSP}\\[1mm]
% \text{final exponent} & \min\{\gamma,h(q_\eta)\} & \min\{\gamma_\eta,\frac12 h(\theta_\eta)\}
% \end{array}
% \]

% The point is that, in Case~1, the successful-set exponent can be kept free of any explicit dependence on the normalized optimum scale, whereas in Case~2 such dependence is already unavoidable in the upper bound on the conditioning set. This is what makes the local-Lipschitz framework natural in the second regime.

% ==================================================
% Section 5: Technical details
% ==================================================
\section{Technical 
Details}\label{sec:technical-details}

This section collects the standard estimates and calculations deferred from Section~\ref{sec:favorable-sets}.

% ==================================================
\subsection{Details of the correlated-pair analysis}

The goal of this subsection is to prove \Cref{th:ns-favorable}. We will need the following three lemmas.

\begin{lemma}[Typical-shell estimates]\label{lem:ns-shell-estimates}
There exists a constant \(C_{k,\eta}>0\) such that, for all sufficiently large \(n\), the following hold.
\begin{enumerate}
    \item The size of the typical shell satisfies
    \[
    2^{h(q_\eta)n-C_{k,\eta}n^{2/3}}
    \le
    |A_\eta|
    \le
    2^{h(q_\eta)n+C_{k,\eta}n^{2/3}}.
    \]
    \item For any point \(x\in A_\eta\), 
    \[
    2^{-h(q_\eta)n-C_{k,\eta}n^{2/3}}
    \le
    \Pr[t=x]
    \le
    2^{-h(q_\eta)n+C_{k,\eta}n^{2/3}}
    \]
    holds
    under the product distribution with parameter \(q_{\eta,n}\).
    \item The complement of the typical shell has exponentially small probability mass  under the product distribution with parameter \(q_{\eta,n}\):
    \[
    \Pr[t\notin A_\eta]\le 2e^{-2n^{1/3}}.
    \]
\end{enumerate}
\end{lemma}

\begin{proof}
Let
\[
q_n=q_{\eta,n},
\qquad
R=\wt(t).
\]
Then \(R\sim\mathrm{Bin}(n,q_n)\). Hoeffding's inequality gives
\[
\Pr[\abs{R-q_n n}>n^{2/3}]
\le
2\exp\!\left(-\frac{2n^{4/3}}{n}\right)
=
2e^{-2n^{1/3}},
\]
which proves Item 3.

Since \(q_n\to q_\eta\in(0,1/2)\), there exists a compact interval
\[
I\subset (0,1/2)
\]
such that, for all sufficiently large \(n\), both \(q_n\) and every ratio \(r/n\) with
\[
|r-q_n n|\le n^{2/3}
\]
belong to \(I\).

Fix such an integer \(r\). Stirling's formula yields
\[
\binom{n}{r}=2^{nh(r/n)+O(\log n)},
\]
uniformly for \(r/n\in I\). Also,
\[
\frac{r}{n}=q_n+O(n^{-1/3}),
\]
so by smoothness of \(h\) on \(I\),
\[
h(r/n)=h(q_n)+O_{k,\eta}(n^{-1/3}).
\]
Since
\[
q_n=q_\eta+O_{k,\eta}(1/n),
\]
we further have
\[
h(q_n)=h(q_\eta)+O_{k,\eta}(1/n).
\]
Therefore
\[
\binom{n}{r}=2^{h(q_\eta)n+O_{k,\eta}(n^{2/3})}.
\]
Summing over the \(O(n^{2/3})\) admissible values of \(r\) gives Item 1.

Let \(x\in A_\eta\) and write \(r=\wt(x)\). Then
\[
\Pr[t=x]=q_n^r(1-q_n)^{n-r}=2^{-nf_n(r/n)},
\]
where
\[
f_n(\alpha)=-\alpha\log_2 q_n-(1-\alpha)\log_2(1-q_n).
\]
The functions \(f_n\) are smooth on \(I\), uniformly in \(n\), and \(r/n=q_n+O(n^{-1/3})\), so
\[
f_n(r/n)=f_n(q_n)+O_{k,\eta}(n^{-1/3}).
\]
But \(f_n(q_n)=h(q_n)=h(q_\eta)+O_{k,\eta}(1/n)\), hence
\[
\Pr[t=x]=2^{-h(q_\eta)n+O_{k,\eta}(n^{2/3})},
\]
which proves Item 2.
\end{proof}

\begin{lemma}[Lower bound on the successful set in the correlated-pair regime]\label{lem:ns-successful-lower}
For all sufficiently large \(n\),
\[
|S_\eta^{\mathrm{ns}}|
\ge
2^{h(q_\eta)n-o(n)}.
\]
\end{lemma}

\begin{proof}
Let
\[
B_\eta=\{t\in\{-1,1\}^n:H(x^\ast\odot t)\le (1-\eta)H_{\mathrm{min}}\}.
\]
By \Cref{prop:ns-tail}, for the choice of \(q=q_{\eta,n}\),
%By \Cref{prop:ns-tail}, with \(\rho_{\eta,n}^k=1-\eta+\eta/n\),
\[
\Pr[B_\eta]\ge \frac1n.
\]
By item 3 of \Cref{lem:ns-shell-estimates},
\[
\Pr[A_\eta^c]\le 2e^{-2n^{1/3}},
\]
and thus
\[
\Pr[A_\eta\cap B_\eta]\ge \frac{1}{2n}
\]
for all sufficiently large \(n\).

By item 2 of \Cref{lem:ns-shell-estimates}, every point of \(A_\eta\) has probability at most
\[
2^{-h(q_\eta)n+o(n)}.
\]
Therefore
\[
|A_\eta\cap B_\eta|
\ge
\frac{1}{2n}\cdot 2^{h(q_\eta)n-o(n)}
=
2^{h(q_\eta)n-o(n)}.
\]
Since coordinatewise multiplication by \(x^\ast\) is a bijection on \(\{-1,1\}^n\), this proves the claim.
\end{proof}

\begin{lemma}[Search-ball volume in the correlated-pair regime]\label{lem:ns-ball-volume}
For all sufficiently large \(n\) and all \(x\in\{-1,1\}^n\),
\[
|\Ball{x}{r_\eta^{\mathrm{ns}}}|
\le
2^{h(q_\eta)n+o(n)}.
\]
\end{lemma}

\begin{proof}
The ball volume depends only on the radius, so it suffices to bound
\[
\sum_{j=0}^{r_\eta^{\mathrm{ns}}}\binom{n}{j}.
\]
Since
\[
\frac{r_\eta^{\mathrm{ns}}}{n}=q_\eta+O(n^{-1/3}),
\]
and \(q_\eta<1/2\), the binomial coefficients are increasing on the range \(0\le j\le r_\eta^{\mathrm{ns}}\) for all sufficiently large \(n\). Hence
\[
\sum_{j=0}^{r_\eta^{\mathrm{ns}}}\binom{n}{j}
\le
(r_\eta^{\mathrm{ns}}+1)\binom{n}{r_\eta^{\mathrm{ns}}}.
\]
Stirling's formula gives
\[
\binom{n}{r_\eta^{\mathrm{ns}}}\le 2^{h(q_\eta)n+o(n)},
\]
and the prefactor \(r_\eta^{\mathrm{ns}}+1\le n+1\) contributes only \(2^{o(n)}\).
\end{proof}

We are now ready to collect all these results and prove \Cref{th:ns-favorable}.
\begin{proof}[Proof of \Cref{th:ns-favorable}]
Item 1 (the inclusion \(S_\eta^{\mathrm{ns}}\subseteq T_\eta\)) is immediate from the definition.

If \(x=x^\ast\odot t\in S_\eta^{\mathrm{ns}}\), then
\[
d_H(x,x^\ast)=\wt(t).
\]
Since \(t\in A_\eta\),
\[
\wt(t)\le q_{\eta,n}n+n^{2/3}=q_\eta n+O_{k,\eta}(1)+n^{2/3}\le r_\eta^{\mathrm{ns}}
\]
for all sufficiently large \(n\). This proves Item 2.

Item 3 and Item 4 are \Cref{lem:ns-successful-lower} and \Cref{lem:ns-ball-volume}, respectively.
\end{proof}

% ==================================================
\subsection{Details of the local-Lipschitz analysis}

The goal of this subsection is to prove \Cref{th:lip-favorable}. We will need the following three lemmas.

The first lemma gives a lower bound on the number of coordinates of low degree. The proof follows from a straightforward counting argument.
\begin{lemma}[Many light coordinates]\label{lem:lip-half-light}
$|L_{\av}|\ge \frac n2.$
\end{lemma}
% \begin{lemma}[Many light coordinates]\label{lem:lip-half-light}
% % Let
% % \[
% % L_{\le 2d_{\mathrm{avg}}}=\{i\in[n]: d_i\le 2d_{\mathrm{avg}}\}.
% % \]
% % Then
% $
% |L_{\av}|\ge \frac n2.
% $
% \end{lemma}

% % \begin{proof}
% % If more than \(n/2\) coordinates had weighted degree strictly larger than \(2d_{\mathrm{avg}}\), then the average weighted degree would exceed \(d_{\mathrm{avg}}\), contradicting
% % \[
% % \frac1n\sum_{i=1}^n d_i=d_{\mathrm{avg}}.
% % \]
% % \end{proof}

The second lemma is a standard bound on the binomial coefficients (see, e.g., \cite[Lemma~10.2]{MitzenmacherUpfal2017}).
\begin{lemma}[Entropy lower bound for a layer]\label{lem:lip-layer-entropy}
Let \(N\ge 1\) and \(0\le r\le N\) be integers. Then
\[
\binom{N}{r}\ge \frac{1}{N+1}\,2^{Nh(r/N)}.
\]
\end{lemma}

% \begin{proof}
% If \(r=0\), the claim is immediate. Otherwise set \(p=r/N\in(0,1/2]\). Under the binomial distribution Binomial\((N,p)\), the probability of any string of Hamming weight exactly \(r\) is
% \[
% p^r(1-p)^{N-r}=2^{-Nh(p)}.
% \]
% Hence
% \[
% \Pr[\text{weight}=r]
% =
% \binom{N}{r}2^{-Nh(p)}.
% \]
% Since \(r=Np\) is a mode of the Binomial\((N,p)\) distribution and there are only \(N+1\) possible weights,
% \[
% \Pr[\text{weight}=r]\ge \frac1{N+1}.
% \]
% Rearranging gives the claim.
% \end{proof}

The next lemma is a variant in which $r$ is expressed as $r=\lfloor Nt\rfloor$ for some parameter $t$.
\begin{lemma}[Rounded binomial lower bound]\label{lem:lip-rounded-binomial}
Let \(N\ge 1\), \(t\in[0,1/2]\), and set
$
r=\lfloor Nt\rfloor.
$
Then
\[
%\sum_{j=0}^{r}\binom{N}{j}
%\ge
\binom{N}{r}
\ge
\frac{1}{eN(N+1)}\,2^{Nh(t)}.
\]
\end{lemma}

\begin{proof}
If \(t=0\), then \(r=0\) and the claim is trivial. We therefore assume that \(t>0\).

If \(Nt<1\), then again \(r=0\) and \(\binom{N}{r}=1\). Since
\[
h(u)\le u\log_2\!\frac{e}{u}
\qquad (u\in(0,1])
\]
and 
\(t<1/N\), we have
\[
Nh(t)\le Nt\log_2\!\frac{e}{t}\le \log_2(eN).
\]
Therefore
\[
\frac{1}{eN(N+1)}\,2^{Nh(t)}
\le
\frac{1}{eN}\,2^{Nh(t)}
\le
1
=\binom{N}{r},
\]
which proves the claim.

Now suppose \(Nt\ge 1\). Then
\[
\frac{r}{N}\in\left[t-\frac1N,t\right]\subseteq\left[\frac1N,\frac12\right].
\]
Since
\[
h'(x)=\log_2\!\frac{1-x}{x},
\]
we have \(0\le h'(x)\le \log_2 N\) on \([1/N,1/2]\). By the mean value theorem,
\[
h(r/N)\ge h(t)-\frac{\log_2 N}{N}.
\]
Applying \Cref{lem:lip-layer-entropy} gives
\[
\binom{N}{r}
\ge
\frac{1}{N+1}\,2^{Nh(r/N)}
\ge
\frac{1}{N+1}\,2^{Nh(t)-\log_2 N}
=
\frac{1}{N(N+1)}\,2^{Nh(t)},
\]
which is stronger than claimed.
\end{proof}

We are now ready to prove \Cref{th:lip-favorable}.
%%%%%%%%%%%%%%%%%%%%%%%%%%%%%%%%%%%%%%%%%%%%%%%%%%
%%%%%%%%%%%%%%%%%%%%%%%%%%%%%%%%%%%%%%%%%%%%%%%%%%
%%%%%%%%%%%%%%%%%%%%%%%%%%%%%%%%%%%%%%%%%%%%%%%%%%
\begin{proof}[Proof of \Cref{th:lip-favorable}]
We give the proof of the four items.

\paragraph{Item 1.}
The inclusion \(S_\eta^{\mathrm{lip}}\subseteq T_\eta\) follows immediately from \Cref{prop:lip-local}: for every \(x\in S_\eta^{\mathrm{lip}}\),
\[
H(x)\le H_{\mathrm{min}}+2\Lambda_{\max} d_{\mathrm{avg}}\,r_\eta^{\mathrm{lip}}
\le
H_{\mathrm{min}}+\eta|H_{\mathrm{min}}|
=
(1-\eta)H_{\mathrm{min}}.
\]

\paragraph{Item 2.}
If \(x\in S_\eta^{\mathrm{lip}}\), then by definition
\[
d_H(x,x^\ast)\le r_\eta^{\mathrm{lip}}
\quad\text{and}\quad
\{i:x_i\neq x_i^\ast\}\subseteq L_{\av}.
\]
Therefore,
\[
x^\ast\in \LightBall{\av}{x}{r_\eta^{\mathrm{lip}}},
\]
as claimed.

\paragraph{Item 3.}
We prove that
\[
|S_\eta^{\mathrm{lip}}|
\ge
2^{\frac12 h(\theta_\eta)n-o(n)}.
\]

Choose any subset \(L'\subseteq L_{\av}\) of size
\[
N=\lfloor n/2\rfloor.
\]
This is possible by \Cref{lem:lip-half-light}. Since
\[
r_\eta^{\mathrm{lip}}
=
\left\lfloor \frac{\theta_\eta n}{2}\right\rfloor,
\]
we distinguish two cases.

\smallskip
\noindent\emph{Case 1: \(\theta_\eta\le 1/2\).}
In this case,
\[
r_\eta^{\mathrm{lip}}
=
\left\lfloor \frac{\theta_\eta n}{2}\right\rfloor
\ge
\lfloor \theta_\eta N\rfloor,
\]
because \(N=\lfloor n/2\rfloor\le n/2\). Hence
\[
\LightBall{L'}{x^\ast}{\lfloor \theta_\eta N\rfloor}\subseteq S_\eta^{\mathrm{lip}},
\]
and therefore
\begin{align*}
|S_\eta^{\mathrm{lip}}|
&\ge
\sum_{j=0}^{\lfloor \theta_\eta N\rfloor}\binom{N}{j}\\
&\ge
\binom{N}{\lfloor \theta_\eta N\rfloor}\\
&\ge
2^{Nh(\theta_\eta)-o(n)}\\
&=
2^{\frac12 h(\theta_\eta)n-o(n)},
\end{align*}
where the third line follows from \Cref{lem:lip-rounded-binomial}.

\smallskip
\noindent\emph{Case 2: \(\theta_\eta> 1/2\).}
Then
\[
r_\eta^{\mathrm{lip}}
=
\left\lfloor \frac{\theta_\eta n}{2}\right\rfloor
\ge
\left\lfloor \frac{n}{4}\right\rfloor
\ge
\left\lfloor \frac{N}{2}\right\rfloor.
\]
Hence
\[
\LightBall{L'}{x^\ast}{\lfloor N/2\rfloor}\subseteq S_\eta^{\mathrm{lip}},
\]
and thus
\begin{align*}
|S_\eta^{\mathrm{lip}}|
&\ge
\binom{N}{\lfloor N/2\rfloor}\\
&\ge
2^{N-o(n)}\\
&=
2^{\frac n2-o(n)}.
\end{align*}
Since \(\theta_\eta\in(1/2,1]\) in this case and the binary entropy satisfies
\[
h(\theta_\eta)\le 1,
\]
we obtain
\[
|S_\eta^{\mathrm{lip}}|
\ge
2^{\frac12 h(\theta_\eta)n-o(n)}.
\]

Thus, in all cases,
\[
|S_\eta^{\mathrm{lip}}|
\ge
2^{\frac12 h(\theta_\eta)n-o(n)}.
\]

% \paragraph{Item 3.}
% We first note that
% \[
% 0\le \theta_\eta\le 1.
% \]
% Indeed, \(\eta\in(0,1)\), \(|H_{\mathrm{min}}|\le W\), and
% \[
% \Sigma=\sum_{j=1}^m k_j w_j \ge \sum_{j=1}^m w_j = W,
% \]
% since each constraint is nontrivial and hence has arity at least \(1\). Therefore
% \[
% \theta_\eta
% =
% \frac{\eta |H_{\mathrm{min}}|}{\Lambda_{\max}\Sigma}
% \le
% \frac{\eta W}{\Lambda_{\max}\Sigma}
% \le 1.
% \]

% To obtain a lower bound on \(|S_\eta^{\mathrm{lip}}|\), choose any subset \(L'\subseteq L_{\av}\) of size
% \[
% N=\lfloor n/2\rfloor.
% \]
% This is possible by \Cref{lem:lip-half-light}. Since
% \[
% r_\eta^{\mathrm{lip}}
% =
% \left\lfloor \frac{\theta_\eta n}{2}\right\rfloor
% \ge
% \lfloor \theta_\eta N\rfloor,
% \]
% we have
% \[
% \LightBall{L'}{x^\ast}{\lfloor \theta_\eta N\rfloor}\subseteq S_\eta^{\mathrm{lip}},
% \]
% and therefore
% \begin{align*}
% |S_\eta^{\mathrm{lip}}|
% &\ge
% \sum_{j=0}^{\lfloor \theta_\eta N\rfloor}\binom{N}{j}\\
% &\ge
% \binom{N}{\lfloor \theta_\eta N\rfloor}.
% \end{align*}

% If \(\theta_\eta\le 1/2\), \Cref{lem:lip-rounded-binomial} gives
% \[
% \binom{N}{\lfloor \theta_\eta N\rfloor}
% \ge
% 2^{\frac12 h(\theta_\eta)n-o(n)}.
% \]

% If \(\theta_\eta>1/2\), then by symmetry of the binary entropy function and the monotonicity of Hamming-ball volume up to radius \(N/2\),
% \[
% \sum_{j=0}^{\lfloor \theta_\eta N\rfloor}\binom{N}{j}
% \ge
% \sum_{j=0}^{\lfloor N/2\rfloor}\binom{N}{j}
% \ge
% 2^{N-o(N)}
% \ge
% 2^{\frac12 h(\theta_\eta)n-o(n)},
% \]
% since \(h(\theta_\eta)\le 1\). Thus in all cases
% \[
% |S_\eta^{\mathrm{lip}}|
% \ge
% 2^{\frac12 h(\theta_\eta)n-o(n)}.
% \]

\paragraph{Item 4.}
The search region
\[
\LightBall{\av}{x}{r_\eta^{\mathrm{lip}}}
\]
has the same cardinality for every center \(x\), and by definition that cardinality is exactly \(|S_\eta^{\mathrm{lip}}|\).
\end{proof}
\section{Concluding Remarks}\label{sec:concluding}

\paragraph{Generality of the framework.}
The abstract formulation of \Cref{th:abstract-conditioning-search} was designed not to rely on uniform sampling from the conditioning set, and can therefore also be applied to more structured distributions, such as those produced by Gibbs samplers or Markov chains. Such distributions were used in \cite{Chakrabarti-etal24} to construct quantum algorithms for Max Bisection, Max Independent Set, and for finding the ground states of the Antiferromagnetic Ising Model and the Sherrington-Kirkpatrick Model. Although we do not pursue this systematically in this paper, our approach may therefore be extended to dequantize these quantum algorithms as well. 

At a broader level, we believe that the conditioning-and-search viewpoint may also be useful for the design of classical algorithms beyond the dequantization setting considered here. In an idealized setting, access to an exact proximity (or distance) test to an optimum would allow one to condition directly on a Hamming ball and then search within that region, in a way reminiscent of Sch\"{o}ning-type algorithms~\cite{Schoning02}. Our threshold sets should be viewed as efficiently testable coarse surrogates for such unavailable distance tests. From this perspective, the framework developed here may be relevant not only for comparing classical and quantum algorithms, but also as a possible organizing principle for the design of new exact classical algorithms.

\paragraph{Barriers and hard instances.}
The present results also point to concrete barriers.

% In the instance-determined regime, which in this paper is instantiated by weighted MAX-\(k\)-CSP,
% %of arity at most \(k\), 
% the gain in our comparison with \Cref{th7} deteriorates when the intrinsic normalized optimum scale \(|H_{\mathrm{min}}|/\Sigma\) is small or the irregularity parameter \(D\) is large. Thus, instances with \(|H_{\mathrm{min}}|/\Sigma=o(1)\) or \(D=\omega(1)\) naturally resist the successful-set construction of \Cref{th:lip-favorable} arising from the local-Lipschitz analysis. This identifies a concrete class of hard instances for our approach. In other words, this is a robustness barrier: the local-Lipschitz mechanism fails to produce a sufficiently large explicit successful set.

In the instance-determined regime, which in this paper is instantiated by MAX-\(k\)-CSP,
%of arity at most \(k\), 
the gain in our comparison with \Cref{th7} deteriorates when the normalized optimum scale $\Delta$ is small or the irregularity parameter \(D\) is large. Thus, instances with $\Delta=o(1)$ or \(D=\omega(1)\) naturally resist the successful-set construction of \Cref{th:lip-favorable} arising from the local-Lipschitz analysis. This identifies a concrete class of hard instances for our approach. In other words, this is a robustness barrier: the local-Lipschitz mechanism fails to produce a sufficiently large explicit successful set.

More generally, in the given-\((\gamma,\eta)\) regime, which in this paper is instantiated by MAX-E\(k\)-LIN2, our method converts the lower bound from \Cref{prop:ns-tail} into the explicit shell-counting statement proved in \Cref{lem:ns-successful-lower}. If the thin-threshold condition fails, then random sampling already yields a good approximate solution efficiently; the remaining difficulty is exact optimization. This suggests that progress beyond our improvement on \Cref{th6} may require more global counting-based methods for exact optimization, perhaps closer in spirit to Williams-style exact exponential-time algorithms~\cite{Williams05}.

\paragraph{Dequantization and complexity.}
The paper suggests that an important part of the short-path speedup story can be dequantized, but that the source of the resulting classical gain need not be uniform across settings: in our comparisons, it is entropy-driven for MAX-E\(k\)-LIN2 and robustness-driven for MAX-\(k\)-CSP. It would be interesting to understand how far this approach can be pushed, how it interacts with other forms of dequantization, and how it relates to fine-grained complexity hypotheses such as SETH and QSETH~\cite{AaronsonCLWZ20,BuhrmanPS21,ChiaHLGS26,ImpagliazzoP01,ImpagliazzoPZ01}. In this sense, the present paper should be viewed less as an endpoint than as a starting point for a broader study of conditioning, explicit successful sets, and exact optimization in both classical and quantum settings.

% ==================================================
\section*{AI Disclosure}
We used ChatGPT 5.4 Thinking Extended and Gemini 3.1 Pro to assist with detecting potential logical flaws and inconsistencies in the formal mathematical content throughout the manuscript, as well as brainstorming throughout various stages of the research and writing process. The authors verified the correctness and originality of all content, including the validity of all mathematical steps and the accuracy of all references.
%We used [Tool Name] to assist with [Brief Description of Use]. The tool materially affected [Sections X and Y]. More details can be found in [Section Z]. The authors verified the correctness and originality of all content including references.

\section*{Acknowledgments}%
The authors are grateful to Alexander Dalzel, Matthew Hastings and Ryan Williams for helpful correspondence.
FLG is supported by JSPS KAKENHI grants JP24H00071 and JP25K24674, MEXT Q-LEAP grant JPMXS0120319794, JST ASPIRE grant JPMJAP2302 and JST CREST grant JPMJCR24I4.  
ST is supported by JSPS KAKENHI grants JP20H05961, JP20H05967, and JP22K11909.

%=====================================
\bibliographystyle{alpha}
\bibliography{main}

\appendix
\section{Removing the Known-Optimum Assumption}\label{sec:standard-variants}

The formal algorithmic statements in the main text, namely \Cref{th:concrete-eklin2,th:concrete-maxkcsp}, are stated under the known-optimum assumption, i.e., under the assumption that the optimum value \(H_{\mathrm{min}}\) is given. We keep this formulation in the main text because it makes the role of the successful sets completely explicit and significantly simplifies the presentation. As we will show in this appendix (and in Appendices \ref{sec:top-k-variant} and \ref{sec:optimistic-first-local-ball}), this assumption can be removed. As already mentioned in Remarks~\ref{rem:1} and~\ref{rem:2}, this is trivial in the unweighted case. We explain here how to handle the weighted case.

%First, for the algorithms stated under the known-optimum assumption in the main text, the expected-time formulations immediately yield standard bounded-error variants once one makes the relevant success probabilities explicit. Thus the known-optimum statements may be converted into bounded-error formulations with the same running-time exponent.

Note that in the weighted case, removing the known-optimum assumption is not merely a routine binary-search reduction. The reason is that the threshold
\[
(1-\eta)H_{\mathrm{min}}
\]
does not appear in isolation: it is tied to the success argument itself. In particular, one does not simply have a monotone yes/no decision oracle whose positive instances can be located by binary search. A wrong guess can fail in two different ways: it may define the wrong sublevel set, and it may also invalidate the local-search step built on that threshold information.

The paper treats two regimes. In the given-\((\gamma,\eta)\) regime, realized through the correlated-pair analysis of \Cref{subsec:noise-stability-analysis}, success is certified by a suitable typical Hamming layer. In the parameter-derived \((\gamma,\eta)\) regime, realized through the local-ball analysis of \Cref{subsec:local-lipschitz-analysis}, success is certified by a suitable restricted Hamming ball around an optimum assignment. %The corresponding variants without the known-optimum assumption are possible, but the reductions needed to remove that assumption are logically downstream from the successful-set constructions and would substantially interrupt the flow of the main text if presented there.
The two regimes differ in exactly how the above difficulty appears.

In the given-\((\gamma,\eta)\) regime, realized in the MAX-E\(k\)-LIN2 setting of \Cref{subsec:noise-stability-analysis}, the search radius
\[
r_\eta^{\mathrm{ns}}=\left\lceil q_\eta n+2n^{2/3}\right\rceil
\]
depends only on \(\eta\), \(k\), and \(n\). Here the issue created by the unknown optimum is independent of the search radius: the only missing ingredient is the explicit membership test
\[
H(x)\le (1-\eta)H_{\mathrm{min}}.
\]
Appendix~\ref{sec:top-k-variant} handles this by a ranking-based reduction, which replaces explicit thresholding by retaining the sampled points with the smallest \(H\)-values. The key input is that the successful set is large enough to be hit with high probability, while the near-optimal region is thin enough that successful points are not discarded by the ranking step. This yields a bounded-error algorithm with the same running-time exponent as in \Cref{th:concrete-eklin2}, up to the same \(2^{o(n)}\) overhead already present there.

% In the given-\((\gamma,\eta)\) regime, realized in the MAX-E\(k\)-LIN2 setting of \Cref{subsec:noise-stability-analysis}, the search radius
% \[
% r_\eta^{\mathrm{ns}}=\left\lceil q_\eta n+2n^{2/3}\right\rceil
% \]
% depends only on \(\eta\), \(k\), and \(n\). Here the issue created by the unknown optimum is geometry-independent: the only missing ingredient is the explicit membership test
% \[
% H(x)\le (1-\eta)H_{\mathrm{min}}.
% \]
% Section~\ref{sec:top-k-variant} handles this by a ranking-based reduction, which replaces explicit thresholding by retaining the lowest-energy sampled points. The key input is that the successful set is large enough to be hit with high probability, while the near-optimal region is thin enough that successful points are not discarded by the Top-\(K\) truncation. This yields a bounded-error algorithm with the same exponential rate as in \Cref{th:concrete-eklin2}, up to the same \(2^{o(n)}\) slack already present there.

In the parameter-derived \((\gamma_\eta,\eta)\) regime, realized in the weighted MAX-\(k\)-CSP setting of \Cref{subsec:local-lipschitz-analysis}, the search radius itself depends on the unknown value \(H_{\mathrm{min}}\), since
% In the parameter-derived \((\gamma,\eta)\) regime, realized in the weighted MAX-\(k\)-CSP setting of \Cref{subsec:local-lipschitz-analysis}, the search radius itself depends on the unknown value \(H_{\mathrm{min}}\), since
\[
r_\eta^{\mathrm{lip}}
=
\left\lfloor
\frac{\eta |H_{\mathrm{min}}|}{2\Lambda_{\max} d_{\mathrm{avg}}}
\right\rfloor.
\]
Thus the issue created by the unknown optimum also affects the choice of search radius: one must guess a value large enough to make the local-ball argument go through, while preserving the correct running-time exponent. Appendix~\ref{sec:optimistic-first-local-ball} handles this by a one-sided reduction over candidate upper bounds. The key point is that invalid guesses fail safely, while the running-time exponent is monotone in the guessed bound. This yields, for each fixed rational \(\eta\in(0,1)\), a high-probability success-time guarantee with the same running-time exponent as the corresponding known-optimum statement in \Cref{th:concrete-maxkcsp}: with high probability, the algorithm outputs an optimum assignment within that running time.

\section[Removing the Known-Optimum Assumption in the Given-$(\gamma,\eta)$ Regime]{Removing the Known-Optimum Assumption in the Given-\texorpdfstring{$\boldsymbol{(\gamma,\eta)}$}{(gamma,eta)} Regime}
\label{sec:top-k-variant}

In the given-\((\gamma,\eta)\) regime, the sublevel-set exponent is supplied externally, and the only missing ingredient in the unknown-optimum setting is the explicit threshold based on \(H_{\mathrm{min}}\). In the correlated-pair framework of \Cref{subsec:noise-stability-analysis}, the search radius is already available independently of \(H_{\mathrm{min}}\), since
\[
r_\eta^{\mathrm{ns}}=\left\lceil q_\eta n+2n^{2/3}\right\rceil
\]
depends only on \(\eta\), \(k\), and \(n\). We therefore isolate a ranking-based reduction that removes the need to know the threshold explicitly.
% In the given-\((\gamma,\eta)\) regime, the sublevel-set exponent is supplied externally, and the only missing ingredient in the unknown-optimum setting is the explicit threshold based on \(H_{\mathrm{min}}\). In the correlated-pair framework of \Cref{subsec:noise-stability-analysis}, the recovery geometry is already available independently of \(H_{\mathrm{min}}\), since the search radius
% \[
% r_\eta^{\mathrm{ns}}=\left\lceil q_\eta n+2n^{2/3}\right\rceil
% \]
% depends only on \(\eta\), \(k\), and \(n\). We therefore isolate a ranking-based reduction that removes the need to know the threshold explicitly.

Throughout this appendix, \(\eta\in(0,1)\) and \(\gamma>0\) are fixed explicit parameters independent of \(n\). The only quantity not assumed to be known is the optimum value \(H_{\mathrm{min}}\).

\begin{theorem}[Generic ranking-based reduction]
\label{thm:generic-top-k}
Let \(\Omega\) be a finite search space, and let \(H\colon \Omega\to \mathbb R\) be an objective function to minimize. Suppose there exist subsets \(S\subseteq T\subseteq \Omega\) and parameters \(s,t,M,\tau\) such that:
\begin{enumerate}
    \item \(S \subseteq \{x\in\Omega:H(x)\le \tau\}\subseteq T\);
    \item for every \(x\in S\), there is a region \(\mathcal N(x)\subseteq \Omega\) of size at most \(M\) that contains at least one optimum point of \(H\);
    \item \(|S|\ge s\) and \(|T|\le t\).
\end{enumerate}
Let \(\delta\in(0,1/2)\), and define
\[
N=\left\lceil \frac{|\Omega|}{s}\ln\frac{2}{\delta}\right\rceil
\qquad\text{and}\qquad
K=\min\!\left\{N,\left\lceil \frac{2}{\delta}N\frac{t}{|\Omega|}\right\rceil\right\}.
\]
Then the algorithm that samples \(N\) points uniformly from \(\Omega\), retains the \(K\) sampled points with smallest \(H\)-value, exhaustively searches the region of each retained point, and outputs a point of minimum \(H\)-value among all examined points, outputs an optimum point with probability at least \(1-\delta\). Its running time is
\[
O\!\left(N\cdot \mathrm{poly}(n)+K M\cdot \mathrm{poly}(n)\right).
\]
% Then the algorithm that samples \(N\) points uniformly from \(\Omega\), retains the \(K\) points of lowest \(H\)-value, exhaustively searches the region of each retained point, and outputs a minimum-energy point among all examined points, outputs an optimum point with probability at least \(1-\delta\). Its running time is
% \[
% O\!\left(N\cdot \mathrm{poly}(n)+K M\cdot \mathrm{poly}(n)\right).
% \]
\end{theorem}

\begin{proof}
Let \(X_S\) and \(X_T\) be the numbers of sampled points in \(S\) and \(T\), respectively. The probability that the sample misses \(S\) entirely is
\[
(1-|S|/|\Omega|)^N \le \exp(-Ns/|\Omega|)\le \delta/2.
\]
Since
\[
\mathbb E[X_T]=N|T|/|\Omega|\le Nt/|\Omega|,
\]
Markov's inequality gives
\[
\Pr[X_T>K]\le \delta/2
\]
(trivially if \(K=N\)). Hence, with probability at least \(1-\delta\), the sample contains some \(x_0\in S\) and at most \(K\) points from \(T\).

On this event, every sampled point \(z\) with \(H(z)\le H(x_0)\) lies in \(T\), since \(H(x_0)\le \tau\) and \(\{x:H(x)\le \tau\}\subseteq T\). Thus at most \(K\) sampled points have \(H\)-value at most \(H(x_0)\), so \(x_0\) is among the retained points. By assumption 2, \(\mathcal N(x_0)\) contains an optimum point. Since the algorithm exhaustively searches \(\mathcal N(x_0)\) and finally outputs a point of minimum \(H\)-value among all examined points, it outputs an optimum point.
% On this event, every sampled point \(z\) with \(H(z)\le H(x_0)\) lies in \(T\), since \(H(x_0)\le \tau\) and \(\{x:H(x)\le \tau\}\subseteq T\). Thus at most \(K\) sampled points have energy at most \(H(x_0)\), so \(x_0\) is among the retained points. By assumption 2, \(\mathcal N(x_0)\) contains an optimum point. Since the algorithm exhaustively searches \(\mathcal N(x_0)\) and finally outputs a minimum-energy point among all examined points, it outputs an optimum point.

The running-time bound follows immediately.
\end{proof}

Here is the main result of this appendix.
\begin{theorem}[Removing the known-optimum assumption in the given-\texorpdfstring{$(\gamma,\eta)$}{(gamma,eta)} regime]
\label{th:top-k-without-hmin}
Fix \(\eta\in(0,1)\) and \(\gamma>0\). Assume
\[
|T_\eta|\le 2^{(1-\gamma)n},
\qquad
T_\eta=\{x:H(x)\le (1-\eta)H_{\mathrm{min}}\}.
\]
For all sufficiently large \(n\), even if \(H_{\mathrm{min}}\) is unknown
there is a randomized algorithm that  outputs an optimum assignment with probability at least \(9/10\), in time
\[
2^{\left(1-\min\{\gamma,h(q_\eta)\}\right)n+o(n)}.
\]
\end{theorem}

\begin{proof}
We instantiate \Cref{thm:generic-top-k} with
\[
\Omega=\{-1,1\}^n,
\qquad
\tau=(1-\eta)H_{\mathrm{min}},
\qquad
S=S_\eta^{\mathrm{ns}},
\qquad
T=T_\eta.
\]
By construction,
\[
S_\eta^{\mathrm{ns}}\subseteq \{x:H(x)\le \tau\}\subseteq T_\eta.
\]
By \Cref{th:ns-favorable}, every \(x\in S_\eta^{\mathrm{ns}}\) satisfies
\[
x^\ast\in \Ball{x}{r_\eta^{\mathrm{ns}}},
\]
and
\[
\max_x |\Ball{x}{r_\eta^{\mathrm{ns}}}|
\le 2^{h(q_\eta)n+O(n^{2/3})}.
\]
Also,
\[
|S_\eta^{\mathrm{ns}}|\ge 2^{h(q_\eta)n-O(n^{2/3})},
\]
while the assumption gives
\[
|T_\eta|\le 2^{(1-\gamma)n}.
\]
Substituting these bounds into \Cref{thm:generic-top-k} with \(\delta=1/10\) gives the claimed success probability.

The sample size satisfies
\[
N = \left\lceil \frac{2^n}{|S_\eta^{\mathrm{ns}}|}\ln 20\right\rceil
= 2^{(1-h(q_\eta))n+O(n^{2/3})}.
\]
The retention cap is
\[
K\le \left\lceil 20N\frac{|T_\eta|}{2^n}\right\rceil
\le 2^{(1-h(q_\eta)-\gamma)n+O(n^{2/3})}.
\]
Thus, the search cost is bounded by
\[
K\cdot \max_x |\Ball{x}{r_\eta^{\mathrm{ns}}}|
\le
2^{(1-\gamma)n+O(n^{2/3})}.
\]
Since \(O(n^{2/3})=o(n)\), the total running time is
\[
2^{(1-h(q_\eta))n+o(n)} + 2^{(1-\gamma)n+o(n)}
=
2^{(1-\min\{\gamma,h(q_\eta)\})n+o(n)},
\]
as desired.
\end{proof}

\section[{Removing the Known-Optimum Assumption in the Parameter-Derived $(\gamma_\eta,\eta)$ Regime}]{Removing the Known-Optimum Assumption in the Parameter-Derived \texorpdfstring{$\boldsymbol{(\gamma_\eta,\eta)}$}{(gamma_eta,eta)} Regime}
% \section[{Removing the Known-Optimum Assumption in the Parameter-Derived $(\gamma,\eta)$ Regime}]{Removing the Known-Optimum Assumption in the Parameter-Derived \texorpdfstring{$\boldsymbol{(\gamma,\eta)}$}{(gamma,eta)} Regime}
\label{sec:optimistic-first-local-ball}

In the parameter-derived \((\gamma_\eta,\eta)\) regime, the threshold exponent is not supplied externally but is instead obtained from the instance parameters.
% In the parameter-derived \((\gamma,\eta)\) regime, the threshold exponent is not supplied externally but is instead obtained from the instance parameters. 
In the local-Lipschitz framework of \Cref{subsec:local-lipschitz-analysis}, this same dependence also enters the search radius, since
\[
r_\eta^{\mathrm{lip}}
=
\left\lfloor
\frac{\eta |H_{\mathrm{min}}|}{2\Lambda_{\max} d_{\mathrm{avg}}}
\right\rfloor
\]
depends on the unknown value \(H_{\mathrm{min}}\). Thus the issue is not only that the threshold is unknown: the search radius itself depends on \(H_{\mathrm{min}}\).

The basic idea is to try candidate upper bounds \(U\) for \(H_{\mathrm{min}}\), starting from the most optimistic ones. For each such \(U\), we run a one-sided routine: invalid guesses must fail safely, while valid guesses should return an optimum with high probability. The remaining point is to check that the running-time exponent is monotone in the guessed upper bound.

% In the parameter-derived \((\gamma,\eta)\) regime, the threshold exponent is not supplied externally but is instead obtained from the instance parameters. In the local-Lipschitz framework of \Cref{subsec:local-lipschitz-analysis}, this same dependence also enters the recovery geometry, since the search radius
% \[
% r_\eta^{\mathrm{lip}}
% =
% \left\lfloor
% \frac{\eta |H_{\mathrm{min}}|}{2\Lambda_{\max} d_{\mathrm{avg}}}
% \right\rfloor
% \]
% depends on the unknown optimum value \(H_{\mathrm{min}}\). Thus the issue is not only that the threshold is unknown: the recovery step itself depends on the unknown optimum scale.

% The basic idea is to try candidate upper bounds \(U\) for \(H_{\mathrm{min}}\), starting from the most optimistic ones. For each such \(U\), we run a one-sided routine: invalid guesses must fail safely, while valid guesses should recover an optimum with high probability. The remaining point is to check that the running-time exponent behaves monotonically with the guessed scale.

Throughout this appendix, \(\eta\in\mathbb Q\cap(0,1)\) is fixed independently of \(n\). The only quantity not assumed to be known is \(H_{\mathrm{min}}\).

We state the next theorem in an abstract form. In our application to weighted MAX-\(k\)-CSP, \(OPT=H_{\mathrm{min}}\), and a ``solution'' means an optimum assignment. The word “deterministically” emphasizes the one-sided nature of the reduction: if the guessed bound is too optimistic, the routine always returns \(\mathsf{NULL}\), independently of its internal randomness.
% \begin{theorem}[Generic one-sided reduction over candidate bounds]
% \label{thm:generic-optimistic}
% Let \(OPT\) be an unknown target value. Suppose we are given candidate bounds
% \[
% U_R < U_{R-1} < \cdots < U_1
% \]
% and a randomized routine \(\mathrm{SearchBounded}(U,\delta)\) such that:
% \begin{enumerate}
%     \item \textbf{One-sidedness:} If \(U < OPT\), the routine deterministically returns \(\mathsf{NULL}\). If \(U \ge OPT\), it outputs a valid solution to the target search problem with probability at least \(1-\delta\).
%     \item \textbf{Monotone cost:} The worst-case running time is \(T(U_r)\), where \(T(U_r)\) is nonincreasing in \(r\).
%     \item \textbf{Valid regime:} There exists \(r^\ast\in\{1,\dots,R\}\) such that
%     \[
%     U_r<OPT \quad (r>r^\ast),
%     \qquad
%     U_{r^\ast}\ge OPT.
%     \]
% \end{enumerate}
% Consider the algorithm that iterates \(r=R,R-1,\dots,1\), computes
% \[
% y_r\leftarrow \mathrm{SearchBounded}(U_r,\delta),
% \]
% and halts returning \(y_r\) if \(y_r\neq\mathsf{NULL}\). Then, with probability at least \(1-\delta\), the algorithm halts at stage \(r^\ast\), outputs a valid solution, and the total running time up to termination is
% \[
% O(R\cdot T(U_{r^\ast})).
% \]
% \end{theorem}

\begin{theorem}[Generic one-sided reduction over candidate bounds]
\label{thm:generic-optimistic}
Let \(OPT\) be an unknown target value. Suppose we are given candidate bounds
\[
U_R < U_{R-1} < \cdots < U_1
\]
such that there exists \(r^\ast\in\{1,\dots,R\}\) satisfying
\[
U_r<OPT \quad (r>r^\ast),
\qquad
U_{r^\ast}\ge OPT,
\]
and a randomized routine \(\mathrm{SearchBounded}(U,\delta)\) such that:
\begin{enumerate}
    \item \textbf{One-sidedness:} If \(U < OPT\), the routine deterministically returns \(\mathsf{NULL}\). If \(U \ge OPT\), it outputs a solution to the target search problem with probability at least \(1-\delta\).
    \item \textbf{Monotone cost:} The worst-case running time is \(T(U_r)\), where \(T(U_r)\) is nonincreasing in \(r\).
\end{enumerate}
Consider the algorithm that iterates \(r=R,R-1,\dots,1\), computes
\[
y_r\leftarrow \mathrm{SearchBounded}(U_r,\delta),
\]
and halts returning \(y_r\) if \(y_r\neq\mathsf{NULL}\). Then, with probability at least \(1-\delta\), the algorithm halts at stage \(r^\ast\), outputs a solution, and the total running time up to termination is
\[
O(R\cdot T(U_{r^\ast})).
\]
\end{theorem}

\begin{proof}
By assumption on the candidate bounds, there exists \(r^\ast\in\{1,\dots,R\}\) such that
\[
U_r<OPT \quad (r>r^\ast),
\qquad
U_{r^\ast}\ge OPT.
\]

For all \(r>r^\ast\), one has \(U_r<OPT\), so by one-sidedness the routine deterministically returns \(\mathsf{NULL}\). Hence the algorithm cannot halt before stage \(r^\ast\).

At stage \(r^\ast\), one has \(U_{r^\ast}\ge OPT\), so the routine outputs a solution with probability at least \(1-\delta\). On this event, the algorithm halts at stage \(r^\ast\).

Conditioned on this event, the executed stages are exactly \(r=R,R-1,\dots,r^\ast\). By monotonicity,
\[
T(U_r)\le T(U_{r^\ast})
\qquad\text{for all }r\ge r^\ast.
\]
Since at most \(R\) stages are executed, the total running time is \(O(R\cdot T(U_{r^\ast}))\).
\end{proof}

We now instantiate \Cref{thm:generic-optimistic} in the weighted MAX-\(k\)-CSP setting of \Cref{subsec:local-lipschitz-analysis}.

For any explicit rational \(U\in[-W,0)\), define
\[
T_\eta(U)=\{x:H(x)\le (1-\eta)U\},
\qquad
r_\eta(U)=
\left\lfloor
\frac{\eta |U|}{2\Lambda_{\max} d_{\mathrm{avg}}}
\right\rfloor,
\qquad
S_\eta(U)=
\LightBall{\av}{x^\ast}{r_\eta(U)}.
\]
Also define
\[
c^{\mathrm{ball}}_\eta(U)
=
\frac12\,h\!\left(
\frac{\eta |U|}{\Lambda_{\max}\Sigma}
\right),
\]
\[
c^{\mathrm{tail}}_\eta(U)
=
\frac{2(1-\eta)^2}{\ln 2}\cdot
\frac{1}{\Lambda_{\max}^2 D}
\left(\frac{|U|}{\Sigma}\right)^2,
\qquad
c_\eta(U)=\min\{c^{\mathrm{ball}}_\eta(U),c^{\mathrm{tail}}_\eta(U)\}.
\]
% \[
% c^{\mathrm{ball}}_\eta(U)
% =
% \frac12\,h\!\left(
% \frac{\eta}{\Lambda_{\max}k}\cdot \frac{|U|}{W}
% \right),
% \]
% \[
% c^{\mathrm{tail}}_\eta(U)
% =
% \frac{2(1-\eta)^2}{\ln 2}\cdot
% \frac{1}{\Lambda_{\max}^2 k^2D}
% \left(\frac{|U|}{W}\right)^2,
% \qquad
% c_\eta(U)=\min\{c^{\mathrm{ball}}_\eta(U),c^{\mathrm{tail}}_\eta(U)\}.
% \]
Finally, let
\[
c_\eta^\star=c_\eta(H_{\mathrm{min}}).
\]

\begin{proposition}[One-sided bounded-search routine]
\label{prop:upper-bound-routines}
There is a randomized routine \(\mathrm{SearchBounded}(U,\delta)\) taking
\[
U\in[-W,0),
\qquad
\delta\in(0,1/2),
\]
such that:
\begin{enumerate}
    \item If \(H_{\mathrm{min}}\le U\), it outputs an optimum assignment with probability at least \(1-\delta\).
    \item If \(H_{\mathrm{min}}>U\), it deterministically outputs \(\mathsf{NULL}\).
    \item Its worst-case running time is at most
    \[
    2^{(1-c_\eta(U))n+o(n)}.
    \]
\end{enumerate}
\end{proposition}

\begin{proof}
Assume first that \(H_{\mathrm{min}}\le U\). Then the same argument as in \Cref{th:lip-favorable} shows that
\[
S_\eta(U)\subseteq T_\eta(U).
\]
Moreover, replacing \(H_{\mathrm{min}}\) by \(U\) in the successful-set and threshold-set bounds from the arity-at-most-\(k\) local-Lipschitz framework gives
\[
V_U=|S_\eta(U)|\ge \frac{2^{c^{\mathrm{ball}}_\eta(U)n}}{\mathrm{poly}(n)}
\qquad\text{and}\qquad
|T_\eta(U)|\le T_{\max}(U)=2^{(1-c^{\mathrm{tail}}_\eta(U))n}.
\]
% Moreover, replacing \(H_{\mathrm{min}}\) by \(U\) in the successful-set and sublevel-set bounds gives
% \[
% V_U=|S_\eta(U)|\ge \frac{2^{c^{\mathrm{ball}}_\eta(U)n}}{\mathrm{poly}(n)}
% \qquad\text{and}\qquad
% |T_\eta(U)|\le T_{\max}(U)=2^{(1-c^{\mathrm{tail}}_\eta(U))n}.
% \]

The routine works as follows.
Sample
\[
N=
\left\lceil
\frac{2^n}{V_U}\ln\frac{2}{\delta}
\right\rceil
\]
points uniformly from \(\{-1,1\}^n\), and let \(\mathcal K\) be the set of sampled points satisfying
\[
H(x)\le (1-\eta)U.
\]
If
\[
|\mathcal K|>
\left\lceil
\frac{2}{\delta}N\frac{T_{\max}(U)}{2^n}
\right\rceil,
\]
return \(\mathsf{NULL}\). Otherwise, exhaustively search
\[
\LightBall{\av}{x}{r_\eta(U)}
\]
for each \(x\in\mathcal K\), and return the minimum-energy point \(y\) found if \(H(y)\le U\), and \(\mathsf{NULL}\) otherwise.

If \(H_{\mathrm{min}}>U\), then no assignment \(y\) satisfies \(H(y)\le U\), so the final acceptance test always fails. Hence the routine deterministically returns \(\mathsf{NULL}\). This proves Item 2.

Now assume \(H_{\mathrm{min}}\le U\). The probability that the sample misses \(S_\eta(U)\) entirely is at most
\[
\left(1-\frac{V_U}{2^n}\right)^N
\le
\exp\!\left(-N\frac{V_U}{2^n}\right)
\le
\frac{\delta}{2}.
\]
Also, since every point of \(\mathcal K\) lies in \(T_\eta(U)\),
\[
\mathbb E[|\mathcal K|]\le N\frac{T_{\max}(U)}{2^n}.
\]
Hence Markov's inequality gives
\[
\Pr\!\left[
|\mathcal K|>
\frac{2}{\delta}N\frac{T_{\max}(U)}{2^n}
\right]
\le
\frac{\delta}{2}.
\]
Therefore, with probability at least \(1-\delta\), the sample contains some \(x\in S_\eta(U)\) and the cap is not exceeded. On this event,
\[
x^\ast\in \LightBall{\av}{x}{r_\eta(U)},
\]
so the search examines \(x^\ast\). Since
\[
H(x^\ast)=H_{\mathrm{min}}\le U,
\]
the final acceptance test passes, and the returned point is an optimum assignment. This proves Item 1.
% Therefore, with probability at least \(1-\delta\), the sample contains some \(x\in S_\eta(U)\) and the cap is not exceeded. On this event,
% \[
% x^\ast\in \LightBall{L_{\le 2d_{\mathrm{avg}}}}{x}{r_\eta(U)},
% \]
% so the search encounters \(x^\ast\). Since
% \[
% H(x^\ast)=H_{\mathrm{min}}\le U,
% \]
% the final acceptance test passes, and the returned point is an optimum assignment. This proves item 1.

For the running-time bound, evaluating the \(N\) sampled points costs \(N\cdot\mathrm{poly}(n)\). Conditional on not aborting, the number of retained centers is at most
\[
\left\lceil
\frac{2}{\delta}N\frac{T_{\max}(U)}{2^n}
\right\rceil,
\]
and each restricted-ball search costs at most \(V_U\cdot\mathrm{poly}(n)\). Therefore the total running time is at most
\[
N\cdot \mathrm{poly}(n)
+
\left\lceil
\frac{2}{\delta}N\frac{T_{\max}(U)}{2^n}
\right\rceil
V_U\cdot \mathrm{poly}(n).
\]
Substituting
\[
N=
\left\lceil
\frac{2^n}{V_U}\ln\frac{2}{\delta}
\right\rceil
\]
and using the bounds on \(V_U\) and \(T_{\max}(U)\) yields Item 3.
\end{proof}

We now give the main result of this appendix.
Note that \(c_\eta^\star\) is exactly the same exponent as the quantity \(c\) in \Cref{th:concrete-maxkcsp}.
\begin{theorem}[Removing the known-optimum assumption in the parameter-derived \texorpdfstring{$(\gamma,\eta)$}{(gamma,eta)} regime]
\label{thm:optimistic-first-local-ball}
Fix \(\eta\in\mathbb Q\cap(0,1)\). Even if \(H_{\mathrm{min}}\) is unknown, there is a randomized algorithm that, with probability at least \(9/10\), outputs an optimum assignment within time
\[
2^{(1-c_\eta^\star)n+o(n)}.
\]
\end{theorem}

\begin{proof}
Let
\[
B_\eta=\frac{\eta}{2\Lambda_{\max} d_{\mathrm{avg}}},
\qquad
R_\eta=\lfloor B_\eta W\rfloor,
\qquad
U_r=-\frac{r}{B_\eta}\quad (r=1,\dots,R_\eta),
\]
and
\[
r^\ast=\lfloor B_\eta |H_{\mathrm{min}}|\rfloor.
\]

If \(r^\ast=0\), then
\[
|H_{\mathrm{min}}|<\frac{1}{B_\eta}=O(W/n).
\]
In this case \(c_\eta^\star n=O(\log n)\), so exhaustive search in time \(2^n\) is already
\[
2^{(1-c_\eta^\star)n+o(n)}.
\]

Assume from now on that \(r^\ast\ge 1\). We apply \Cref{thm:generic-optimistic} with
\[
OPT=H_{\mathrm{min}}
\]
and stage-\(r\) routine \(\mathrm{SearchBounded}(U_r,1/10)\). By \Cref{prop:upper-bound-routines}, the one-sidedness condition holds.

Since \(|U_r|=r/B_\eta\) increases with \(r\), both
\[
c^{\mathrm{ball}}_\eta(U_r)
\qquad\text{and}\qquad
c^{\mathrm{tail}}_\eta(U_r)
\]
are nondecreasing in \(r\). Hence so is
\[
c_\eta(U_r)=\min\{c^{\mathrm{ball}}_\eta(U_r),c^{\mathrm{tail}}_\eta(U_r)\},
\]
which means that the running-time bound
\[
2^{(1-c_\eta(U_r))n+o(n)}
\]
is nonincreasing in \(r\). This is exactly the monotone-cost condition.
% Since \(|U_r|=r/B_\eta\) increases with \(r\), both
% \[
% c^{\mathrm{ball}}_\eta(U_r)
% \qquad\text{and}\qquad
% c^{\mathrm{tail}}_\eta(U_r)
% \]
% are nondecreasing in \(r\). Hence so is
% \[
% c_\eta(U_r)=\min\{c^{\mathrm{ball}}_\eta(U_r),c^{\mathrm{tail}}_\eta(U_r)\},
% \]
% which means that the running-time bound
% \[
% O^*\!\left(2^{(1-c_\eta(U_r))n}\right)
% \]
% is nonincreasing in \(r\). This is the monotone-cost condition.

By definition of \(r^\ast\),
\[
U_r<H_{\mathrm{min}}
\qquad (r>r^\ast),
\qquad
U_{r^\ast}\ge H_{\mathrm{min}}.
\]
Thus the valid-regime condition also holds.

Applying \Cref{thm:generic-optimistic}, we conclude that, with probability at least \(9/10\), the algorithm halts at stage \(r^\ast\), outputs an optimum assignment, and has total running time
\[
O\!\left(
R_\eta\cdot 2^{(1-c_\eta(U_{r^\ast}))n}\,\mathrm{poly}(n)
\right).
\]
Since \(R_\eta=O(n)\), this is
\[
2^{(1-c_\eta(U_{r^\ast}))n+o(n)}.
\]

It remains to compare \(c_\eta(U_{r^\ast})\) with \(c_\eta^\star=c_\eta(H_{\mathrm{min}})\). Since
\[
r^\ast=\lfloor B_\eta |H_{\mathrm{min}}| \rfloor,
\]
we have
\[
0\le U_{r^\ast}-H_{\mathrm{min}}<\frac{1}{B_\eta}.
\]
Now
\[
B_\eta=\frac{\eta}{2\Lambda_{\max}d_{\mathrm{avg}}}
\qquad\text{and}\qquad
d_{\mathrm{avg}}=\frac{\Sigma}{n},
\]
so
\[
\frac{1}{B_\eta}
=
\frac{2\Lambda_{\max}\Sigma}{\eta n}
=
O(\Sigma/n).
\]
Therefore
\[
\left|
\frac{|U_{r^\ast}|}{\Sigma}
-
\frac{|H_{\mathrm{min}}|}{\Sigma}
\right|
=
O(1/n).
\]
The tail exponent depends quadratically on \(|U|/\Sigma\), so it changes by \(O(1/n)\). The ball exponent has the form
\[
\frac12\,h(a|U|/\Sigma)
\]
for a fixed constant \(a>0\), and since \(h'(x)=O(\log(1/x))\) near \(x=0\), replacing \(|H_{\mathrm{min}}|/\Sigma\) by a quantity within \(O(1/n)\) changes this exponent by at most
\[
O\!\left(\frac{\log n}{n}\right).
\]
Hence
\[
c_\eta(U_{r^\ast})
=
c_\eta^\star+O\!\left(\frac{\log n}{n}\right),
\]
which is absorbed into the \(o(n)\) term in the exponent. This proves the theorem.
% It remains to compare \(c_\eta(U_{r^\ast})\) with \(c_\eta^\star=c_\eta(H_{\mathrm{min}})\). Since
% \[
% r^\ast=\lfloor B_\eta |H_{\mathrm{min}}| \rfloor,
% \]
% we have
% \[
% 0\le U_{r^\ast}-H_{\mathrm{min}}<\frac{1}{B_\eta}=O(W/n),
% \]
% and therefore
% \[
% \left|
% \frac{|U_{r^\ast}|}{W}
% -
% \frac{|H_{\mathrm{min}}|}{W}
% \right|
% =
% O(1/n).
% \]
% The tail exponent depends quadratically on \(|U|/W\), so it changes by \(O(1/n)\). The ball exponent has the form
% \[
% \frac12\,h(a|U|/W)
% \]
% for a fixed constant \(a>0\), and since \(h'(x)=O(\log(1/x))\) near \(x=0\), replacing \(|H_{\mathrm{min}}|/W\) by a quantity within \(O(1/n)\) changes this exponent by at most
% \[
% O\!\left(\frac{\log n}{n}\right).
% \]
% Hence
% \[
% c_\eta(U_{r^\ast})
% =
% c_\eta^\star+O\!\left(\frac{\log n}{n}\right),
% \]
% which is absorbed into the \(o(n)\) term in the exponent. This proves the theorem.
\end{proof}
% ==================================================
% Section D: Notation guide for Sections 2--4
% ==================================================
\section{Notation Guide for Sections~2--4}\label{sec:notation-guide}

This appendix collects the main notation used in Sections~\ref{sec:framework}, \ref{sec:favorable-sets}, and \ref{sec:concrete-corollaries}. It is intended only as a quick reference: definitions, assumptions, and proofs are given in the main text. Some symbols, such as \(T_\eta\), \(S_\eta\), and \(\kappa_\eta\), are reused with section-specific refinements; in each case we indicate the relevant section.

\begingroup
\setlength{\LTleft}{0pt}
\setlength{\LTright}{0pt}
\renewcommand{\arraystretch}{1.18}
\rowcolors{2}{gray!8}{white}

% ==================================================
\subsection*{Section~\ref{sec:framework}: A Two-Parameter Framework for Exact Optimization}

\begin{longtable}{L{0.18\textwidth}L{0.56\textwidth}L{0.18\textwidth}}
\rowcolor{gray!22}
\textbf{Symbol} & \textbf{Meaning} & \textbf{First used} \\
\endfirsthead
\rowcolor{gray!22}
\textbf{Symbol} & \textbf{Meaning} & \textbf{First used} \\
\endhead

\(\Omega_n\) &
Ambient search space in the abstract framework. In Section~\ref{sec:framework}, the black-box theorem later specializes to
\[
\Omega_n=\{-1,1\}^n.
\]
& Section~\ref{sec:framework} \\

\(T\subseteq \Omega_n\) &
Conditioning set in the abstract theorem. &
\Cref{th:abstract-conditioning-search} \\

\(S\subseteq T\) &
Successful subset of the conditioning set: if the sampled point lies in \(S\), the attached search region is guaranteed to contain a target point. &
\Cref{th:abstract-conditioning-search} \\

\(\mu\) &
Sampling distribution supported on \(T\). &
\Cref{th:abstract-conditioning-search} \\

\(f:\Omega_n\to\{0,1\}\) &
Polynomial-time computable target predicate; the goal is to find some \(y\) with \(f(y)=1\). &
\Cref{th:abstract-conditioning-search} \\

\(\mathcal N(x)\) &
Search region attached to a sampled point \(x\). &
\Cref{th:abstract-conditioning-search} \\

\(x^\ast\) &
A target point satisfying \(f(x^\ast)=1\); every \(x\in S\) has \(x^\ast\in \mathcal N(x)\). &
\Cref{th:abstract-conditioning-search} \\

\(\tau_{\mathrm{samp}}(n)\) &
Expected running time of one call to the sampler for \(\mu\). &
\Cref{th:abstract-conditioning-search} \\

\(\tau_{\mathrm{search}}(n)\) &
Worst-case time to exhaustively search \(\mathcal N(x)\). &
\Cref{th:abstract-conditioning-search} \\

\(\beta_n\) &
Lower bound on the successful mass:
\[
\mu(S)\ge \beta_n.
\]
& \Cref{th:abstract-conditioning-search} \\

\(T_\eta\) &
Near-optimal threshold set:
\[
T_\eta=\{x:H(x)\le (1-\eta)H_{\mathrm{min}}\}.
\]
& Section~\ref{sec:framework} \\

\(S_\eta\subseteq T_\eta\) &
Successful subset used in the black-box theorem. &
Section~\ref{sec:framework} \\

\(\mathcal N_\eta(x)\) &
Search region attached to \(x\in T_\eta\) in the black-box theorem. &
Section~\ref{sec:framework} \\

\(\gamma_\eta\) &
Threshold exponent defined through an upper bound of the form
\[
|T_\eta|\le 2^{(1-\gamma_\eta)n}.
\]
& Section~\ref{sec:framework} \\

\(\kappa_\eta\) &
Successful-set exponent defined through a lower bound of the form
\[
|S_\eta|\ge 2^{\kappa_\eta n-o(n)}.
\]
& Section~\ref{sec:framework} \\

\(c_\eta\) &
Resulting exact exponent from the black-box theorem:
\[
c_\eta=\min\{\gamma_\eta,\kappa_\eta\}.
\]
& \Cref{th:black-box} \\
\end{longtable}

% ==================================================
\subsection*{Section~\ref{sec:favorable-sets}: Analysis of Successful Sets}

\begin{longtable}{L{0.18\textwidth}L{0.56\textwidth}L{0.18\textwidth}}
\rowcolor{gray!22}
\textbf{Symbol} & \textbf{Meaning} & \textbf{First used} \\
\endfirsthead
\rowcolor{gray!22}
\textbf{Symbol} & \textbf{Meaning} & \textbf{First used} \\
\endhead

\(h(t)\) &
Binary entropy function:
\[
h(t)=-t\log_2 t-(1-t)\log_2(1-t)
\qquad (0\le t\le 1).
\]
& Section~\ref{sec:favorable-sets} \\

\(\kappa_\eta^{\mathrm{ns}}\) &
Successful-set exponent in the correlated-pair regime:
\[
\kappa_\eta^{\mathrm{ns}}=h(q_\eta).
\]
& Section~\ref{sec:favorable-sets} \\

\(\kappa_\eta^{\mathrm{lip}}\) &
Successful-set exponent in the local-Lipschitz regime:
\[
\kappa_\eta^{\mathrm{lip}}=\frac12\,h(\theta_\eta).
\]
& Section~\ref{sec:favorable-sets} \\

\multicolumn{3}{l}{\textit{Correlated-pair / MAX-E\(k\)-LIN2 regime}} \\

\(\mathcal F\) &
Family of \(k\)-subsets indexing the degree-\(k\) monomials of
\[
H(x)=\sum_{S\in\mathcal F} c_S\prod_{i\in S}x_i.
\]
& \Cref{subsec:noise-stability-analysis} \\

\(c_S\) &
Coefficient of the monomial indexed by \(S\in\mathcal F\). &
\Cref{subsec:noise-stability-analysis} \\

\(t=(t_1,\dots,t_n)\) &
Random sign vector used to generate a correlated copy of \(x^\ast\). &
\Cref{prop:ns-contraction} \\

\(q\) &
Bit-flip probability in the correlated-pair construction:
\[
\Pr[t_i=-1]=q,\qquad \Pr[t_i=1]=1-q.
\]
& \Cref{prop:ns-contraction} \\

\(\rho\) &
Correlation parameter:
\[
\rho=\mathbb E[t_i]=1-2q.
\]
& \Cref{prop:ns-contraction} \\

\(X=x^\ast\odot t\) &
\(\rho\)-correlated partner of the optimum assignment \(x^\ast\). &
\Cref{prop:ns-contraction} \\

\(q_{\eta,n}\) &
Finite-\(n\) flip rate:
\[
q_{\eta,n}=\frac{1-\left(1-\eta+\frac{\eta}{n}\right)^{1/k}}{2}.
\]
& \Cref{subsec:noise-stability-analysis} \\

\(q_\eta\) &
Limiting flip rate:
\[
q_\eta=\frac{1-(1-\eta)^{1/k}}{2}.
\]
& \Cref{subsec:noise-stability-analysis} \\

\(\wt(t)\) &
Hamming weight of the \(-1\) coordinates of \(t\):
\[
\wt(t)=|\{i\in[n]:t_i=-1\}|.
\]
& \Cref{subsec:noise-stability-analysis} \\

\(A_\eta\) &
Typical Hamming shell:
\[
A_\eta=
\Bigl\{t\in\{-1,1\}^n:
|\wt(t)-q_{\eta,n}n|\le n^{2/3}
\Bigr\}.
\]
& \Cref{subsec:noise-stability-analysis} \\

\(S_\eta^{\mathrm{ns}}\) &
Successful set in the correlated-pair regime:
\[
S_\eta^{\mathrm{ns}}
=
\{x^\ast\odot t:\ t\in A_\eta,\ H(x^\ast\odot t)\le (1-\eta)H_{\mathrm{min}}\}.
\]
& \Cref{subsec:noise-stability-analysis} \\

\(r_\eta^{\mathrm{ns}}\) &
Search radius in the correlated-pair regime:
\[
r_\eta^{\mathrm{ns}}=\left\lceil q_\eta n+2n^{2/3}\right\rceil.
\]
& \Cref{subsec:noise-stability-analysis} \\

\multicolumn{3}{l}{\textit{Local-Lipschitz / weighted MAX-\(k\)-CSP regime}} \\

\(\vars(j)\) &
Set of variables involved in constraint \(j\). &
\Cref{subsec:local-lipschitz-analysis} \\

\(k_j\) &
Arity of constraint \(j\):
\[
k_j=|\vars(j)|\le k.
\]
& \Cref{subsec:local-lipschitz-analysis} \\

\(w_j\) &
Positive weight of constraint \(j\). &
\Cref{subsec:local-lipschitz-analysis} \\

\(P_j\) &
Local predicate of constraint \(j\). &
\Cref{subsec:local-lipschitz-analysis} \\

\(s_j\) &
Number of satisfying local assignments for \(P_j\):
\[
s_j=|P_j^{-1}(1)|.
\]
& \Cref{subsec:local-lipschitz-analysis} \\

\(C_j(x)\) &
Centered contribution of constraint \(j\) to the objective \(H(x)\). &
\Cref{subsec:local-lipschitz-analysis} \\

\(W\) &
Total constraint weight:
\[
W=\sum_{j=1}^m w_j.
\]
& \Cref{subsec:local-lipschitz-analysis} \\

\(d_i\) &
Weighted degree of variable \(i\):
\[
d_i=\sum_{j:\,i\in\vars(j)} w_j.
\]
& \Cref{subsec:local-lipschitz-analysis} \\

\(\Sigma\) &
Total incident weight:
\[
\Sigma=\sum_{i=1}^n d_i=\sum_{j=1}^m k_j w_j.
\]
& \Cref{subsec:local-lipschitz-analysis} \\

\(d_{\mathrm{avg}}\) &
Average weighted degree:
\[
d_{\mathrm{avg}}=\frac{\Sigma}{n}.
\]
& \Cref{subsec:local-lipschitz-analysis} \\

\(D\) &
Irregularity parameter:
\[
D=\frac{n\sum_{i=1}^n d_i^2}{\Sigma^2}.
\]
& \Cref{subsec:local-lipschitz-analysis} \\

\(\Lambda_j\) &
Local Lipschitz factor of constraint \(j\):
\[
\Lambda_j=\frac{2^{k_j}}{2^{k_j}-s_j}.
\]
& \Cref{subsec:local-lipschitz-analysis} \\

\(\Lambda_{\max}\) &
Worst-case Lipschitz factor:
\[
\Lambda_{\max}=\max_{j\in[m]}\Lambda_j.
\]
& \Cref{subsec:local-lipschitz-analysis} \\

\(\LightBall{L}{x}{r}\) &
Restricted Hamming ball:
\[
\LightBall{L}{x}{r}
=
\{y:\{i:x_i\neq y_i\}\subseteq L,\ d_H(x,y)\le r\}.
\]
& \Cref{subsec:local-lipschitz-analysis} \\

\(L_{\av}\) &
Set of light variables:
\[
L_{\av}=\{i\in[n]:d_i\le 2d_{\mathrm{avg}}\}.
\]
& \Cref{subsec:local-lipschitz-analysis} \\

\(\theta_\eta\) &
Normalized local-ball parameter:
\[
\theta_\eta
=
\frac{\eta |H_{\mathrm{min}}|}{\Lambda_{\max}\Sigma}.
\]
& \Cref{subsec:local-lipschitz-analysis} \\

\(r_\eta^{\mathrm{lip}}\) &
Search radius in the local-Lipschitz regime:
\[
r_\eta^{\mathrm{lip}}
=
\left\lfloor
\frac{\eta |H_{\mathrm{min}}|}{2\Lambda_{\max} d_{\mathrm{avg}}}
\right\rfloor
=
\left\lfloor \frac{\theta_\eta n}{2}\right\rfloor.
\]
& \Cref{subsec:local-lipschitz-analysis} \\

\(S_\eta^{\mathrm{lip}}\) &
Successful set in the local-Lipschitz regime:
\[
S_\eta^{\mathrm{lip}}
=
\LightBall{\av}{x^\ast}{r_\eta^{\mathrm{lip}}}.
\]
& \Cref{subsec:local-lipschitz-analysis} \\
\end{longtable}

% ==================================================
\subsection*{Section~\ref{sec:concrete-corollaries}: Proof of the Main Results}

\begin{longtable}{L{0.18\textwidth}L{0.56\textwidth}L{0.18\textwidth}}
\rowcolor{gray!22}
\textbf{Symbol} & \textbf{Meaning} & \textbf{First used} \\
\endfirsthead
\rowcolor{gray!22}
\textbf{Symbol} & \textbf{Meaning} & \textbf{First used} \\
\endhead

\(\gamma\) &
Given threshold exponent in Case~1:
\[
|T_\eta|\le 2^{(1-\gamma)n}.
\]
& \Cref{subsec:concrete-eklin2} \\

\(\gamma_\eta\) &
Derived threshold exponent in Case~2. In the weighted MAX-\(k\)-CSP setting of arity at most \(k\),
\[
\gamma_\eta
=
\frac{2(1-\eta)^2}{\ln 2}\cdot
\frac{1}{\Lambda_{\max}^2 D}
\left(\frac{|H_{\mathrm{min}}|}{\Sigma}\right)^2.
\]
& \Cref{prop:tech-mcdiarmid} \\

\(\theta_\eta\) &
Local-ball parameter imported from Section~\ref{sec:favorable-sets}; it determines the successful-set exponent
\[
\kappa_\eta^{\mathrm{lip}}=\frac12 h(\theta_\eta).
\]
& \Cref{th:concrete-maxkcsp} \\

\(q_\eta\) &
Correlated-pair parameter imported from Section~\ref{sec:favorable-sets}; it determines the successful-set exponent
\[
\kappa_\eta^{\mathrm{ns}}=h(q_\eta).
\]
& \Cref{th:concrete-eklin2} \\

\(\Delta\) &
Normalized optimum scale used in the Introduction for comparison with \cite{DalzellPCB23}:
\[
\Delta=\frac{|H_{\mathrm{min}}|}{m}
\qquad\text{(unweighted case)}.
\]
& Introduction \\

\(T_\eta\) &
Near-optimal threshold set used in the concrete corollaries:
\[
T_\eta=\{x:H(x)\le (1-\eta)H_{\mathrm{min}}\}.
\]
& Section~\ref{sec:concrete-corollaries} \\

\(S_\eta^{\mathrm{ns}}\), \(S_\eta^{\mathrm{lip}}\) &
Successful sets supplied by Section~\ref{sec:favorable-sets} and used as inputs to the black-box theorem. &
Section~\ref{sec:concrete-corollaries} \\
\end{longtable}

\endgroup

\end{document}